%% file: main.tex
\documentclass[runningheads]{llncs}
\usepackage[T1]{fontenc}
\usepackage{graphicx}
\usepackage{cite}
 
\usepackage{amsmath,amssymb,amsfonts}
\usepackage{algorithmic}
\usepackage{float}
\usepackage{graphicx}
\usepackage{subcaption}
\usepackage{textcomp}
\usepackage{xcolor}
\def\BibTeX{{\rm B\kern-.05em{\sc i\kern-.025em b}\kern-.08em
    T\kern-.1667em\lower.7ex\hbox{E}\kern-.125emX}}
\usepackage{braket}
\usepackage{stmaryrd}
\usepackage{mathtools}
\usepackage{tikzit}
\usepackage{hyperref}
\usepackage[table]{xcolor}
\usepackage[english]{babel}
\usepackage{amsthm}
\usepackage{todonotes}

\input{zx.tikzstyles}
\hypersetup{
    colorlinks=true,
    linkcolor=blue,
    linkbordercolor=cyan,
    filecolor=magenta,      
    urlcolor=black,
    urlbordercolor=cyan,
    citecolor=black,
    citebordercolor=cyan,
    pdfpagemode=FullScreen,
}
\definecolor{zx_red}{RGB}{227, 145, 145}
\definecolor{zx_green}{RGB}{230, 250, 230}
\definecolor{eb}{rgb}{0.03, 0.57, 0.82}
\newcommand\emulates{\mathrel{\overset{\makebox[0pt]{\mbox{\normalfont\small\sffamily e}}}{\mapsfrom}}}
\newcommand{\eqbc}[1]{\stackrel{\mathclap{\normalfont\small\mbox{$#1$}}}{=}}
\newcommand{\diag}{\text{diag}}
\input{qcircuit_common.tex}
\newcommand{\tikzfigscale}[2]{\ensuremath{\vcenter{\hbox{\scalebox{#1}{$\input{#2.tikz}$}}}}}
\DeclarePairedDelimiter{\ceil}{\lceil}{\rceil}

\newcommand{\minu}{\texttt{-}}

\begin{document}
\title{Transversal AND in Quantum Codes}
%
%
\author{Christine Li\inst{1}\orcidID{0000-0001-6443-7090} \and
Lia Yeh \inst{2}\orcidID{0000-0003-2704-4057}}
\authorrunning{C. Li and L. Yeh}
%
\institute{Department of Computer Science, Columbia University, New York, USA \and
Department of Computer Science and Technology, University of Cambridge, Cambridge, UK
\newline \email{cl4315@columbia.edu}\ \ \email{ly404@cam.ac.uk}}
\maketitle              
\begin{abstract}
The AND gate is not reversible---on qubits. However, it \textit{is} reversible on qutrits, making it a building block for efficient simulation of qubit computation using qutrits.
We first observe that there are multiple two-qutrit Clifford+T unitaries that realize the AND gate with T-count 3, and its generalizations to $n$ qubits with T-count $3n-3$.
Our main result is the construction of a novel qutrit $\llbracket 6, 2, 2\rrbracket$ quantum error-correcting code with a transversal implementation of the AND gate.
The key insight in our approach is that a symmetric T-depth one circuit decomposition --- composed of a CX circuit, T and T dagger gates, followed by the CX circuit in reverse --- of a given unitary can be interpreted as a CSS code. We can increase the code distance by augmenting the code circuit with additional stabilizers while preserving the logical gate. This results in a code with a ``built-in'' transversal implementation of the original unitary, which can be further concatenated to attain a $\llbracket 48,2,4 \rrbracket$ code with the same transversal logical gate.
Furthermore, we present several protocols for mixed qubit-qutrit codes which we call Qubit Subspace Codes, and for magic state distillation and injection.

\end{abstract}

\section{Introduction}
When computing with bits, the first gates we often learn are AND, OR, and NOT. When computing with qubits, one of the first concepts that we often learn is that the AND gate is not reversible.
However, there \textit{is} a reversible quantum AND gate---just not for qubits. We escape this restriction specific to qubits by working with the simplest possible extension---qutrits, i.e. three-level quantum systems instead of two-level. 

Qutrits open up promising avenues due to many advantages, including more efficient algorithms, gates, and error-correcting codes. In many existing hardware architectures for quantum computers, qubits are a choice of two-dimensional subspace in what is natively a higher-dimensional setting. In choosing to actively encode  and control quantum information in more levels, quantum computing with qudits is an active area of research in both theory and practice. Qudit computation has been experimentally demonstrated on systems such as nuclear magnetic resonance~\cite{GedikZ2015quditalgperm}, nuclear spins~\cite{GodfrinC2017quditnuclearspins}, cold atoms~\cite{AndersonBE2015quditcoldatoms}, Rydberg atoms~\cite{WeggemansJ2022quditRydberg}, linear optics~\cite{HuXM2018ququart,lysaght_2024}, integrated photonics~\cite{ChiY2022quditphotonicprocessor}, trapped ion~\cite{ringbauer_universal_2022,HrmoP2023quditions}, and superconducting~\cite{goss_high-fidelity_2022,CaoS2023ququartemuvqe}, with a recent demonstration of GKP qutrit error correction beyond break-even~\cite{brock_quantum_2025}.
Qudit entanglement has demonstrated improved noise robustness compared to qubits~\cite{Srivastav_2022}, and leveraging qutrit states can improve fidelity of qubit operations~\cite{BaekkegaardT2019scqubitqutrit,BrownT2022qutritqubitentangle}.
As the most well-studied qudit setting, qutrits can efficiently simulate qubit computation by reducing circuit depth or ancilla overhead, as shown for diagonal qubit unitaries~\cite{wetering_building_2023}, 
logarithmic-depth Toffoli~\cite{gokhale_asymptotic_2019}, and binary AND gates~\cite{chu_scalable_2023}, as well as factoring algorithms~\cite{bocharov_factoring_2017}, probing lattice gauge theories~\cite{Meth2025,joshi_2025_lattice}, variational quantum algorithms~\cite{Cao_2024, lysaght_2024}, and unitary decomposition~\cite{lysaght_2024}. Qutrit operations can also be used to suppress leakage errors to the $\ket{2}$ state~\cite{miao_overcoming_2023}.

However, there is a gap between the qudit operations realizable on real-world hardware, and those that conceivably satisfy definitions of fault-tolerance; while definitions vary, the main principles are unchanged from what Gottesman laid out in 1997~\cite{Gottesman_1998}.
A large reason for this is that pairwise two-level operations in a three-level quantum system, while natural to do in experimental implementations, does not embed naturally into qutrit stabilizer theory the way that it does for qubits. Another reason is that the qubit stabilizer fragment does not embed into the qutrit stabilizer fragment.
Qutrit computation also offers a promising path to more efficient error-correcting codes, with better thresholds and yield rates in magic state distillation~\cite{campbell_magic-state_2012}. While prior work has explored various constructions of qutrit, prime-dimensional qudit, and composite-dimensional~\cite{gunderman2025beyond} codes, less attention has been given to finding qutrit codes with specific transversal logical operators, especially ones that are both non-Clifford and entangling. In this work, we take a logical operator-centric approach to code construction by reversing the typical procedure for designing stabilizer codes. Instead of deriving the logical operators from the fixed stabilizer set of a code, we start with a chosen logical operator, constructing the code around it by adding stabilizers.
This work constitutes one of few efforts on improving qubit computation framed with respect to qutrit stabilizer computation, and to our knowledge is the first to construct qutrit stabilizer error-correcting codes expressedly to this end.

Our approach to constructing qutrit codes with a binary AND gate logical operation centers on two components. The first is exact circuit synthesis, leveraging the known connection between exact circuit synthesis and error-correcting code constructions~\cite{campbell_smallest_2016,CampbellET2017msdsynth}.
The second is the ZX-calculus, by means of which a number of works have investigated new ways to represent and compile for quantum error-correcting codes. However, while more progress has been made on verifying and proving properties of codes and on compiling fault-tolerant circuits, there has been comparatively less use of ZX-calculus to construct specific new codes and protocols.

Combining the motivating ideas of (1) emulating qubit gates using qutrits and (2) constructing a QEC code around a particular logical operator, the central question we set out to answer in this work is: 
\begin{center}\textit{How can we construct a non-trivial qutrit QEC code with a transversal implementation of a qubit logical gate?}
\end{center}

\section{Qubit and Qutrit Reversible Computation}
For qutrits, the computational basis states are $\ket{0}$, $\ket{1}$, and $\ket{2}$. A normalized qutrit state can be written as $\ket{\psi}=\alpha\ket{0}+\beta\ket{1}+\gamma\ket{2}$, where $\alpha,\beta,\gamma\in \mathbb{C}$ and $|\alpha|^2+|\beta|^2+|\gamma|^2=1$. In a generalized $d$-dimensional system with computational basis states $\ket{k}$, where $k\in \mathbb{Z}_d$, the Pauli X and Z gates are defined as
$$X\ket{k}=\ket{k+1\pmod{d}}, \qquad Z\ket{k}=\omega^k\ket{k}$$
where $\omega = e^{\frac{2\pi i}{d}}$. For qutrits, $\omega=e^{\frac{2\pi i}{3}}$, and $\overline{\omega}\coloneq \omega^2 =\omega^{-1}$.
\begin{definition}[Qutrit X gates]
    For qutrits, there are five non-trivial permutations of the computational basis states (all are Clifford): $X_{01},X_{02},X_{12},X_{+1},X_{-1}$. $X_{+1}$ is the Pauli X gate for qutrits. $X_{-1}=X_{+1}^\dagger$, and $X^\dagger = X^2$. $X_{01}$ acts as $\ket{0}\mapsto \ket{1}$, $\ket{1}\mapsto\ket{0}$, and $\ket{2}\mapsto\ket{2}$. $X_{02}$ and $X_{12}$ are defined analogously.
\end{definition}
\begin{definition}[Qutrit X basis]
    The X basis consists of the states $\ket{+}\coloneq \frac{1}{\sqrt{3}}(\ket{0}+\ket{1}+\ket{2})$, $\ket{\omega} \coloneq \frac{1}{\sqrt{3}}(\ket{0}+\omega\ket{1}+\overline{\omega}\ket{2})$, and $\ket{\overline{\omega}}\coloneq \frac{1}{\sqrt{3}}(\ket{0}+\overline{\omega}\ket{1}+\omega\ket{2})$.
\end{definition}
\begin{definition}[Hadamard gate]
    The $H$ gate is acts as $\ket{0}\mapsto \ket{+},\ket{1}\mapsto\ket{\omega},\ket{2}\mapsto\ket{\overline{\omega}}$. Note that unlike for qubits, $H\neq H^\dagger$, but $H^3=H^\dagger$. $H^2 = X_{12}$ is the dualizer (also known as the antipode) for qutrits, and $H^4=I$.
\end{definition}
\begin{definition}[Z phase gate]
    The Z phase gate is $Z(a,b)\coloneq \diag(1,\omega^a,\omega^b)$, where $a,b\in \mathbb{R}$. It acts as $\ket{k} \rightarrow \omega^k \ket{k}$ for $k \in \mathbb{Z}_3$. The Pauli Z gate is $Z(1,2)$.
\end{definition}
\begin{definition}[X phase gate]
    The X phase gate is defined as $X(a,b)\coloneq HZ(a,b)H^\dagger$, where $a,b\in \mathbb{R}$. The Pauli X gate is $X(2,1)$ and $X^\dagger = X(1,2)$.
\end{definition}
\begin{definition}[T gate]
    The T gate is $T\coloneq Z(\frac{1}{3},-\frac{1}{3}) = \diag(1, e^{\frac{2\pi i}{9}},e^{\frac{-2\pi i}{9}})$. Just as in qubits, the qutrit T gate is in the third level of the Clifford hierarchy.
\end{definition}
\noindent For qutrit controlled gates, there are different notions of control. We use the following two in this work:
\begin{definition}[$\Lambda$-controlled $U$]
    Given a qutrit unitary $U$, the $\Lambda$-controlled $U$ gate (defined in \cite{bocharov_factoring_2017}) acts as $\ket{i}\ket{j}\mapsto\ket{i}\otimes (U^i\ket{j})$, for $i,j\in\{0,1,2\}$. An example of this is the two-qutrit CX gate, defined as $\ket{i,j}\mapsto \ket{i,i+j\pmod 3}$. Unlike the CNOT gate for qubits, the CX gate is not self adjoint, but $(CX)^2 = (CX)^\dagger$.
\end{definition}
\begin{definition}[$\ket{0}$-controlled $U$ gate]
    Given a qutrit unitary $U$, the $\ket{0}$-controlled $U$ is defined as $\ket{0}\otimes\ket{\psi}\mapsto \ket{0}\otimes U\ket{\psi}$, $\ket{1}\otimes\ket{\psi}\mapsto\ket{1}\otimes\ket{\psi}$, and $\ket{2}\otimes\ket{\psi}\mapsto \ket{2}\otimes\ket{\psi}$. The $\ket{1}$- and $\ket{2}$-controlled $U$ gates are analogously defined. We can change the control value by conjugating the control qutrit by $X$ or $X^\dagger$.
\end{definition}
\noindent The CX gate is Clifford, while the $\ket{0}$(or $\ket{1}$ or $\ket{2}$-controlled gates are non-Clifford. Similar to the qubit case, the generators of the qutrit Clifford group are $H, S, CX$. Adding the non-Clifford T gate to the set, we get the Clifford+T gate set, which is universal~\cite{cui_universal_2015} and admits an exact circuit synthesis algorithm~\cite{glaudell_exact_2024}.

Like in the qubit case, for any prime qudit dimension, X and Z generate the Pauli group, and the CX, S, and H gates generate the Clifford group.
The $\ket{0}$-controlled gate, or equivalently the CCX gate, when added to the H gate, generates the qudit Toffoli+Hadamard gate set~\cite{Roy_2023}. \cite{yeh_transdimensional_2025}~presents a comprehensive coverage of the Clifford+T and Toffoli+Hadamard gate sets for prime-dimensional qudits, and their representations in qudit ZX- and ZH-calculus.

\subsection{Qubit and Qutrit ZX-calculus}
The ZX-calculus is a graphical formalism for quantum computing in which linear maps are represented by diagrams built from basic generators called Z and X spiders, together with structural elements that allow wires to bend and cross~\cite{Coecke_2011}. Each diagram has a precise interpretation as a linear map: Sequential composition corresponds to matrix multiplication, parallel composition to tensor product, and the number of inputs and outputs determines whether the diagram represents a state, effect, map, or scalar. The calculus is equipped with rewrite rules that equate diagrams representing the same linear map. The qubit ZX-calculus is \emph{sound}---any derivable equality reflects genuine equality of the underlying maps; \emph{universal} for qubit linear maps---any qubit operation can be expressed as a ZX diagram; and \emph{complete}---for any two equal diagrams, their equality can be proved entirely within the calculus~\cite{Hadzihasanovic_2018}.

We do not use the qubit ZX-calculus in this work; we use the qutrit ZX-calculus, first introduced in~\cite{RanchinA2014quditzx, WangQ2014qutritzx}. All its wires represent a qutrit (three-level quantum information carrier) instead of a qubit (two-level).
For qubits~\cite{Backens_2014}, qutrits~\cite{wang_qutrit_2018}, and qudits of any odd prime dimension~\cite{BoothR2022qupitstab,PoorB2023qupitzxtra}, the ZX-calculus is complete for the stabilizer or Clifford fragment of quantum computation.

The qutrit ZX-calculus is generated by the five diagrams 
\begin{equation}
\Biggl\{
\input{qutrit_Z_spider.tikz},\quad
\input{qutrit_X_spider.tikz},\quad
\input{qutrit_H.tikz},\quad
\input{qutrit_swap.tikz},\quad
\input{qutrit_id.tikz}
\Biggr\}
\end{equation}
The interpretations of $Z$ and $X$ spiders as linear maps are
\begin{align}
    \Biggl\llbracket\quad
\mathllap{\scriptstyle m}
\Biggl\{\input{qutrit_Z_spider.tikz}\Biggr\}
\mathrlap{\scriptstyle n}\quad
\Biggr\rrbracket
& \coloneq \ket{0}^{\otimes m}\bra{0}^{\otimes n}+\omega^a \ket{1}^{\otimes m}\bra{1}^{\otimes n} + \omega^b \ket{2}^{\otimes m}\bra{2}^{\otimes n}\\
\Biggl\llbracket\quad
\mathllap{\scriptstyle m}
\Biggl\{\input{qutrit_X_spider.tikz}\Biggr\}
\mathrlap{\scriptstyle n}\quad
\Biggr\rrbracket
& \coloneq \ket{+}^{\otimes m}\bra{+}^{\otimes n}+\omega^a \ket{\omega}^{\otimes m}\bra{\omega}^{\otimes n} + \omega^b \ket{\overline{\omega}}^{\otimes m}\bra{\overline{\omega}}^{\otimes n}
\end{align}
Note that when $m=n=1$, the green Z spider corresponds to the Z phase gate $Z(a,b)$ and the X spider corresponds to the X phase gate $X(a,b)$. We use the non-flex-symmetric version of the qutrit ZX-calculus. This means that, unlike the qubit ZX-calculus, it is no longer the case that only connectivity matters. In particular, CX is not equivalent to CX$^\dagger$:
\begin{equation}
    \tikzfigscale{0.7}{qutrit_CX}\neq \tikzfigscale{0.7}{qutrit_CX_dagger} = \tikzfigscale{0.7}{qutrit_CX_dual} = \tikzfigscale{0.7}{qutrit_CX_dual_box}
\end{equation}
In this work, we will mostly follow the notation from~\cite{townsend-teague_simplification_2022}, where $H$ is the yellow box labeled 1 and $H^\dagger$ is the yellow box labeled 2. However, we use the horizontal orientation (left to right), rather than vertical (bottom to top), in keeping with common ZX-calculus usage in quantum error correction and circuit compilation works. We represent the dualizer (antipode), $X_{12}=H^2$, by a white box labeled $D$. The qutrit ZX-calculus rules are reprinted from~\cite{townsend-teague_simplification_2022} in Figure \ref{fig:townsendteaguequtritrules} in the appendix.

\section{The AND Gate as a Quantum Circuit Primitive}

\subsection{Emulation: Qubit Computation using Qutrits}
The first difficulty with developing a fault-tolerant theory of emulation is that qubit Cliffords are not a subgroup of qudit Cliffords.
As a counterexample to show this: the $\diag(1,-1,1)$ gate is Clifford-equivalent (by conjugating by Pauli X) to the non-Clifford R gate.

We can also see that restricting a qutrit Clifford gate to its qubit subspace, is not necessarily a qubit Clifford operation.
For example, the qutrit Pauli $Z$ gate acts on the qubit subspace as $\diag(1,\omega)$; $\omega = e^{i\frac{2\pi}{3}}$ is not an element of the ring which all qubit Clifford unitary matrices must be over~\cite{GilesB2013multiqubitcliffordplustsynthesis}.

It is impossible to implement binary AND gate using two qubits. However, it \textit{is} possible to do so using two qu\textit{trits}. By this we mean that we can construct two-qutrit circuit such that the behavior (i.e., truth table rows) of the first two levels $\ket{0}$ and $\ket{1}$ of the qutrit circuit matches that of the qubit binary AND gate. The behavior of the third level of the qutrit circuit is left unconstrained. The additional space afforded by the third level allows for the binary AND gate to be implemented unitarily. Let us denote that one linear map $M_1$ emulates another, $M_2$, by $M_2 \emulates M_1$.
\begin{lemma}\label{lemma:binaryand}
    The binary AND gate can be emulated by a two-qutrit Clifford+T unitary with T-count 3.
\end{lemma}
\begin{proof}
    The below Clifford+T two-qutrit unitary emulates the binary AND gate~\cite[Equation 9]{vlasov2015notes}:
    \begin{equation}
        \label{eq:binaryand}
        \tikzfigscale{0.8}{"binary_AND_base_circuit"} \;\emulates \;\;\tikzfigscale{0.8}{qutrit/andsmall}
    \end{equation}
    Its cost is T-count 3 and Clifford CX-count 4.  The cost calculation follows from the $\ket{2}$-controlled $X_{+1}$ gate comprising of 3 Clifford CX gates and 3 qutrit T gates~\cite{bocharov_factoring_2017}.
\end{proof}
\begin{remark}
    After submission of this manuscript, it was brought to our attention that the above construction is near-identical (up to conjugation by the two-qutrit SWAP gate) to that of~\cite[Equation 9]{vlasov2015notes}, which proposed emulation of binary classical reversible computation with ternary classical reversible computation as early as 2011. We thank the author for pointing out their work to us.
\end{remark}

Interestingly, it outperforms the Toffoli gate, for which it is known that the optimal cost is 7 qubit T gates in the no-ancilla case~\cite{GossetD2014opttoffolit}, 6 qubit Clifford CX gates independent of ancillae count~\cite{ShendeVopttoffolicx}, and 5 two-qubit gates~\cite{NengkunY2013opttoffoli2q}.  Even when ancillae and operations conditioned on measurement are permitted, the best known qubit Clifford+T decomposition of the Toffoli gate is 4 T gates~\cite{JonesC2013toffolit}.

A symmetric qutrit circuit which emulates the qutrit binary AND gate is:
\begin{equation}\label{eq:andrev}
    \tikzfigscale{0.8}{"binary_AND_base_circuit"} \;\emulates \;\;\tikzfigscale{0.8}{"qutrit_binary_AND_1"}
\end{equation}
We can check that both these circuits indeed emulate AND through Table~\ref{table:andtruth}.
\begin{table}[h]
\centering
\[\label{thm:anddoubleemu}
\tikzfigscale{0.8}{qutrit/andsmall}
\;
\vcenter{\hbox{{
\setlength{\tabcolsep}{2.5pt}
\renewcommand{\arraystretch}{1.0}
\begin{tabular}{|c|c|>{\columncolor{lightgray}}c|c|}
\hline
a & b & c & a$\land$b \\
\hline
0 & 0 & 0 & 0 \\
\hline
0 & 1 & 2 & 0 \\
\hline
\rowcolor{lightgray}
0 & 2 & 1 & 2 \\
\hline
1 & 0 & 1 & 0 \\
\hline
1 & 1 & 0 & 1 \\
\hline
\rowcolor{lightgray}
1 & 2 & 2 & 1 \\
\hline
\rowcolor{lightgray}
2 & 0 & 2 & 2 \\
\hline
\rowcolor{lightgray}
2 & 1 & 1 & 1 \\
\hline
\rowcolor{lightgray}
2 & 2 & 0 & 2 \\
\hline
\end{tabular}
}}}
\;\;
\vcenter{\hbox{{
\setlength{\tabcolsep}{2.5pt}
\renewcommand{\arraystretch}{1.0}
\begin{tabular}{|c|c|>{\columncolor{lightgray}}c|c|}
\hline
a & b & c & a$\land$b \\
\hline
0 & 0 & 0 & 0 \\
\hline
0 & 1 & 2 & 0 \\
\hline
\rowcolor{lightgray}
0 & 2 & 0 & 2 \\
\hline
1 & 0 & 1 & 0 \\
\hline
1 & 1 & 1 & 1 \\
\hline
\rowcolor{lightgray}
1 & 2 & 0 & 1 \\
\hline
\rowcolor{lightgray}
2 & 0 & 1 & 2 \\
\hline
\rowcolor{lightgray}
2 & 1 & 2 & 1 \\
\hline
\rowcolor{lightgray}
2 & 2 & 2 & 2 \\
\hline
\end{tabular}
\;\;\;
\tikzfigscale{0.8}{qutrit_binary_AND_1}
}}}\]
\caption{Truth tables for AND emulations of Equations~\eqref{eq:binaryand} (left) and~\eqref{eq:andrev} (right).}
\label{table:andtruth}
\end{table}
We can change the control value of the $\ket{2}$-controlled X by conjugating by Pauli X gates, and from controlled X to Z by conjugating by Hadamard gates.
\begin{remark}\label{remark:freedom}
    Through fixing only the rows and columns corresponding to the binary AND gate (i.e., only rows where the inputs are 0 or 1, and the output column a$\land$b), the remaining truth table entries shaded in Table~\ref{table:andtruth} are a degree of freedom. We can therefore explore other circuit constructions so long as the AND action is preserved for binary inputs.
\end{remark}

By applying the X$_{01}$ gate (which emulates the binary NOT gate) to both inputs and to the AND'd output of Equation~\eqref{eq:binaryand}, in accordance with DeMorgan's law, the binary OR gate is emulated with the same T-count and Clifford CX-count.
This leads us to:
\begin{proposition}
    All binary classical logic circuits can be emulated through exact synthesis in the qutrit Clifford+T gate set.
\end{proposition}
\begin{proof}
    AND, OR, NOT is a universal gate set for binary classical logic. As NOT can be emulated by a qutrit Clifford, the T-count is upper bounded by 3 times the total number of AND or OR gates.
\end{proof}

We present in Table~\ref{table:countstable} a summary of the most efficient qutrit emulation of each binary or qubit gate that we calculated in Sections~\ref{sec:XZSCXCZ} to~\ref{sec:mctand} of the appendix. There, we also give constructions for the qubit X, Z, S, CX and CZ gates, and $\ket{+}$ and $\ket{-}$ states, and prove why the qubit H gate cannot be exactly synthesized unitarily in qutrit Clifford+T.  In addition to giving the non-Clifford gate count, for near-term applications we give the number of two-qutrit operations required, as on present-day qubit quantum computers, two-qubit operations are the dominant noise source.  We also give the qutrit Clifford CX count where all other operations are single-qutrit Clifford+T gates.
\definecolor{LightGreen}{rgb}{0.9,0.98,0.9}
\begin{table}[H]
    \centering
    \begin{tabular}{|l|r|r|r|r|c|l|r|r|r|r|}
        \hline
        Gates & T & R & CX & 2Q & & Gates & T & R & CX & 2Q \\
        \hline
        \rowcolor{LightGreen}
        AND/OR & 3 & 0 & 4 & 1 & & $n$-ary AND~\cite{chu_scalable_2023} & $3n-3$ & 0 & $4n-4$ & $n-1$ \\
        X (i.e. NOT) & 0 & 0 & 0 & 0 & & Z & 0 & 1 & 0 & 0 \\
        \rowcolor{LightGreen}
        CX~\cite{bocharov_factoring_2017} & 6 & 0 & 8 & 1 & & CZ~\cite{wetering_building_2023} & 0 & 3 & 3 & 1 \\
        CCX~\cite{bocharov_factoring_2017} & 12 & 0 & 15 & 3 & & CCZ~\cite{wetering_building_2023} & 0 & 3 & 5 & 3 \\
        \rowcolor{LightGreen}
        C$^n$X, linear depth & $6n$ & 0 & $6n+3$ & $2n-1$ & & C$^n$Z, linear depth & $6n-12$ & 3 & $6n-7$ & $2n-1$ \\
        C$^n$X, log depth~\cite{chu_scalable_2023} & $6n$ & 0 & $8n$ & $2n-1$ & & C$^n$Z, log depth & $6n-12$ & 3 & $8n-5$ & $2n-1$ \\
        \hline
    \end{tabular}
    \vspace{0.6em}
    \caption{T gate, R gate, Clifford CX gate, and 2-qutrit gate counts for emulations of qubit and binary gates, unitarily without ancillae in the Clifford+T+R gate set. For decompositions in the literature given in terms of ternary classical reversible gates, these are computed by substituting the best known decompositions gate-wise, optimizing where possible. The $\ket{+}$, $\ket{0}$, and S gate counts are omitted because they are probabilistic constructions.}\label{table:countstable}
\end{table}
These counts are straightforward, but they show that the asymptotic non-Clifford gate counts of $n$-ary AND and qubit C$^n$X in ancilla-free qutrit Clifford+T, outperform the best known constant-ancilla qubit decomposition according to~\cite{dutta2025}: \cite{Khattar_2025} having T-count $8n-12$. The $n$-ary AND moreover has lower T-count than the best known $O(n)$ ancilla qubit decompositions~\cite{gidney2021ccczgateperformed6} having T-count $4n-6$ and T-depth $\ceil{\text{log}_2 \frac{n}{3}} + 2$~\cite{nakanishi2024}.
More importantly, these qutrit decompositions are unitary, meaning that unlike the best-known qubit decompositions, they do not require feedforward of measurement outcomes to classically control any of their gates. We remark the possibility that similar AND gate-based techniques that halved T-count of multiple-controlled Toffolis and other binary computations in the qubit case~\cite{Gidney_2018}, could be used to improve the depth of these qutrit constructions by allowing ancillae, for instance by preceding $n$-ary AND by a depth-1 layer of $n$ CX gates to $n$ ancilla in the $\ket{0}$ state.

Considering evidence that qudit emulation of individual multi-qubit gates in qubit circuits is of limited effectiveness~\cite{lysaght_2024}, instead of resynthesizing gates in qubit circuits on an individual basis, a more promising approach would be to generalize the above approach for the AND gate, to binary classical reversible functions emulated with ternary classical reversible functions.
Explicit constructions exist for any qudit dimension $d$, of any $d$-ary classical reversible function $f:\mathbb{Z}_d^n\to \mathbb{Z}_d^n$ on $n$ dits by a circuit of $O(d^n n)$ $\ket{0}$-controlled $X$ gates and $O(n)$ ancillae prepared in the $\ket{0}$ state (\cite{ZiW2023optsynthmultictrlqudit, Roy_2023}, building on~\cite{yeh_qutrit_ctrl_2022}), thus putting the T-count for any odd prime $d$ at $O(d^n n)$.
This puts the best known upper and lower bounds at a log factor separation: For any qudit dimension $d$, there exist $d$-ary classical reversible functions $f:\mathbb{Z}_d^n\to \mathbb{Z}_d^n$ that require at least $O(n d^n/\log n)$ single-qudit and two-qudit gates to construct, even when allowed $\Omega (n)$ ancillae~\cite{Roy_2023}.
More generally, techniques to compress qubit circuits using qudit decompositions~\cite{lysaght_2024} or using mixed-dimensional quantum systems by clustering weighted graphs~\cite{mato2023compression} can also be leveraged---from there, the difficult next steps are to decompose into conceivably fault-tolerant gates, or to build error-correcting codes around them as we will next do.

\subsection{A T-depth 1 Symmetric Circuit as a \texorpdfstring{$\llbracket 3,2,1 \rrbracket$}{[[3,2,1]]} Qutrit Code}
As noted in Remark~\ref{remark:freedom} there are many qutrit circuits that emulate the binary AND gate. From prioritizing those with low T-count, but more importantly, are symmetric and T-depth one, we found:
\begin{equation}
    \tikzfigscale{0.8}{"qutrit_binary_AND_2"}
\end{equation}
where we can apply the symmetric T-depth one decomposition of the $\ket{0}$-controlled Z gate from Figure 7 of~\cite{bocharov_factoring_2017}.
Therefore, the full Clifford+T symmetric T-depth one qutrit circuit decomposition of the qubit binary AND gate is given by
\begin{equation}\label{eq:boxedcircuit}
    \tikzfigscale{0.78}{circuit_with_boxes}
\end{equation}
In the equivalent ZX-diagram, we see that $E_{pf}$ is a phase-free ZX-diagram:
\begin{equation}\label{eq:symmetricANDzx}
    \tikzfigscale{0.8}{"full_binary_and_ZX_diagram"}
\end{equation}

To see how this forms the encoder (and decoder) of a QEC code, we can use ZX-calculus to transform the first half of circuit into the Generalized Parity Form (GPF). From this we can directly ``read off'' the Z and X-type logical basis operators and stabilizer generators. Notice that this circuit contains other gates besides CNOTs, which means that the corresponding code is not a CSS code. However, for the purposes of finding the stabilizer generators and logical operators and determining code distance, it is sufficient to focus on just the CSS portion of the circuit, which is composed of the inner CX gates~\cite{kissinger_phase-free_2022}. 
Therefore, from the CNOT encoder portion, we get the encoder in Z-X (left) and X-Z (right) normal forms~\cite{kissinger_phase-free_2022,HuangJ2023zxcss,KissingerA2024PQS}:
\begin{equation}
    \tikzfigscale{0.8}{distance_one_Z-X_enc} = \tikzfigscale{0.8}{distance_one_enc} = \tikzfigscale{0.8}{distance_one_X-Z_enc}
\end{equation}
From the connectivity of the Z-X form of the encoder, the X-type logical operators are $\overline{X}_1 = X_1X_3^2$ and $\overline{X}_2 = X_1X_2X_3$. There are no X-type stabilizers. From the X-Z form of the encoder, the Z-type logical operators and stabilizer generators are $\overline{Z}_1 = Z_1Z_2^2$ and $\overline{Z}_2 = Z_2$, and $S_1^{Z} = Z_1Z_2Z_3$. We can observed the code distance is 1, since there is a logical Z-type operator of weight 1.

\section{Reconstructing the Code at Higher Distance}
To increase the code distance, we would need to add X-type stabilizers, while ensuring that the distance stays the same. We take inspiration from an observation about the qubit $\llbracket 8,3,2 \rrbracket$ code construction from~\cite{campbell_smallest_2016}, which we deconstruct in the ZX-calculus in Sections~\ref{sec:circzx} and~\ref{sec:transversalccz} of the appendix.

In an attempt to reproduce how to construct the $\llbracket 8,3,2 \rrbracket$ code by adding an X-type stabilizer to all physical qubits of a $\llbracket 7,3,1 \rrbracket$ code with transversal CCZ, we add an X-type stabilizer that checks all qutrits: $S_1^{X} = X_1 X_2 X_3$. Adding to our earlier encoder, which we conjugate $T^{\otimes 3}$ by, 
and after applying ZX rewrite rules to ``push'' the stabilizer outside of the encoder/decoder, we obtain the following symmetric circuit:
\begin{equation}
    \tikzfigscale{0.8}{"expanded_CZ0_conj_circuit"}
\end{equation}
Since the inner portion of the circuit is still the same as before adding the stabilizer, this is the $\ket{0}$-controlled $Z$ gate with a CX gate applied from an extra ancilla (which is initialized and postselected to $\ket{+}$) to the second data qutrit on either sides. Therefore, the equivalent circuit is
\begin{equation}
    \tikzfigscale{0.8}{"CZ0_conj_circuit"}
\end{equation}
However, we can see that the logical gate corresponding to this circuit is no longer the same.

\subsection{Recovering Logical AND}\label{sec:recoveringAND}
To figure out how to modify the circuit to preserve the logical binary AND gate, let's first see what is the logical gate implemented by this circuit. Since the ancilla qutrit is in the qutrit $\ket{+}$ state just before the first CX gate, we can consider three cases based on the Z basis state of this superposition. For $\ket{0}$, the CX gates controlled by the ancilla does nothing, so the circuit is the same as simply the $\ket{0}$-controlled $Z$ from the first to the second qutrit. The behavior corresponding to this $\ket{0}$ case of the ancilla can be enforced by adding a second control from the ancilla, as in the leftmost case below.
\begin{equation*}
    \begin{matrix}
    \textbf{Ancilla in }\ket{0} \qquad & \qquad \textbf{Ancilla in }\ket{1} \qquad & \qquad \textbf{Ancilla in }\ket{2}\\
    \tikzfigscale{0.8}{"CZ0_conj_ket_0_double_control"} \qquad & \qquad \tikzfigscale{0.8}{"CZ0_conj_ket_1_double_control"} \qquad & \qquad \tikzfigscale{0.8}{"CZ0_conj_ket_2_double_control"}
\end{matrix}
\end{equation*}
For the $\ket{1}$ case of the ancilla, the CX gates apply a $X$ to the left and $X^\dagger$ to the right of $Z$ on the second qutrit. Observing that $X^\dagger ZX = \omega Z$, this simplifies to the middle case above.
Similarly, for the $\ket{2}$ case where two $X$ and $X^\dagger$ gates are applied, since $(X^\dagger)^2 Z X^2 = \omega^2 Z$, we obtain the rightmost case above.

Putting together these three cases, we find that
\begin{equation}
    \tikzfigscale{0.8}{"CZ0_conj_circuit"} = \tikzfigscale{0.8}{"CZ0_conj_circuit_expanded"}
\end{equation}
We can expand out the gates by separating out the $Z$ and the extra phases, then permuting the gates to simplify the circuit:
    \begin{equation}\label{eq:CZ0lemma1}
        \tikzfigscale{0.8}{CZ0_first_lemma_1} = \tikzfigscale{0.8}{CZ0_first_lemma_2} = \tikzfigscale{0.8}{CZ0_first_lemma_3} = \tikzfigscale{0.8}{CZ0_first_lemma_4}
    \end{equation}
Moreover, one can check that
\begin{equation}\label{eq:CZ0lemma2}
        \tikzfigscale{0.8}{CZ0_second_lemma_1} = \tikzfigscale{0.8}{CZ0_second_lemma_2}
    \end{equation}
Putting together Eqs.~\ref{eq:CZ0lemma1} and~\ref{eq:CZ0lemma2}, we can see that the following steps produce a CX-conjugated circuit which still implements the $\ket{0}$-controlled $Z$ gate (the inner portion of the circuit for AND):
\begin{equation}\label{eq:CZ0codecircuit}
    \tikzfigscale{0.8}{"CZ0_base_circuit"} = \tikzfigscale{0.8}{"CZ0_conj_circuit_simplified_conj_again"} = \tikzfigscale{0.8}{CZ0_code_circuit_permuted_last_two_gates} = \tikzfigscale{0.8}{"CZ0_code_circuit"}
\end{equation}
Thus, the final circuit in Eq~\ref{eq:CZ0codecircuit} preserves the logical AND gate and has an additional X-type stabilizer. We will now use this to derive the code.

\subsection{Deriving the Code as a Symmetric, T-depth One Circuit}
In order to determine the corresponding QEC code, our approach is to solve for a symmetric, T-depth one synthesis of the circuit. We can rewrite the circuit in its symmetric form using \emph{phase gadgets}---symmetric multi-qubit interactions native to many hardware architectures~\cite{PinoJM2021honeywellarchitecture,PetitL2022silicontwoqubitgates,SheldonS2016IBMCRgate} and useful for optimizing quantum circuits~\cite{CowtanA2020phasegadget,deBeaudrapN2020reducepifourphase,vandeWeteringJ2021globalgates}.

First, we will find the phase gadgets representation of the $\ket{0}$-controlled $Z$ and $Z^\dagger$ gates. Then, we will find the combined phase gadgets form. From this, we can obtain a T-depth one symmetric form of the circuit.

We begin with the symmetric circuit for $\ket{0}$-controlled from~\cite{bocharov_factoring_2017}. To rewrite this in phase gadgets form, we will need to use Eq (6) of~\cite{wetering_building_2023} (which can be derived using spider fusion, bialgebra rule, and dualizing).
Therefore,
\begin{align}
    &\tikzfigscale{0.7}{"CZ0_base_circuit"} = \tikzfigscale{0.7}{"microsoft_symmetric_CZ0_ZX_diagram"}\eqbc{(f)(s)} \tikzfigscale{0.7}{"symmetric_CZ0_to_phase_gadget_1"}\\
    \;& \eqbc{\text{\cite{wetering_building_2023}}}\;\tikzfigscale{0.8}{"symmetric_CZ0_to_phase_gadget_2"}\;\eqbc{(f)}\;\;\tikzfigscale{0.8}{"symmetric_CZ0_to_phase_gadget_3"} = \tikzfigscale{0.8}{"symmetric_CZ0_to_phase_gadget_4"} = \tikzfigscale{0.8}{"symmetric_CZ0_to_phase_gadget_6"}
\end{align}
To find the phase gadget form of its adjoint, $\ket{0}$-controlled $Z^\dagger$, we take the adjoint of the phase gadget diagram that we just found, by taking its horizontally reflection and negating the phases. Applying dualizers to all wires surrounding the red $X$ spiders flips back the direction of the phase gadgets, so we end up with the same diagram for $\ket{0}$-controlled $Z^\dagger$ except with $T^\dagger$ instead of $T$ phases. Therefore, we can replace the $\ket{0}$-controlled $Z$ and $Z^\dagger$ gates with their respective phase gadget representation, then rewrite and apply the rule from~\cite{wetering_building_2023} to obtain a combined phase gadgets form.
\subsection{Finding the Stabilizers and Logical Operators}
From the combined phase gadgets diagram, we can obtain a symmetric circuit by reading off the connectivity for each of the phase gadgets' spiders. This is because the above diagram corresponds to phase polynomials based on sum over paths for each of the phase gadgets, which can equivalently be expressed in a symmetric form that computes a parity map on the left then uncomputes on the right, which is also equivalent to a CNOT circuit with phase gates and pre/postselection (see Chapter 7.1 of \cite{KissingerA2024PQS}). Thus, from the combined phase gadget diagram above, we obtain the following symmetric forms:
\begin{equation}
  \tikzfigscale{0.8}{"Z-X_enc_dec"} \quad = \quad \tikzfigscale{0.8}{"X-Z_enc_dec_simplified"}
\end{equation}
The left portion of the diagrams represents the encoder maps, and symmetrically, the right portion the decoder map. In the left hand side diagram, the encoder is in the Z-X form~\cite{kissinger_phase-free_2022}. From this, we can read off the X-type stabilizers and logical operators. The diagram on the right hand side is the corresponding X-Z encoder, obtained by applying bialgebra, dual, and fusion rewrite rules. From the X-Z form, we can read off the Z-type stabilizers and logical operators.
Thus, for this inner CSS code (which implements the $\ket{0}$-controlled $Z$), the two equivalent encoder maps, along with the corresponding stabilizer generators and logical operators are:
\begin{equation}\label{eq:Epf}
    \vcenter{
  \hbox{$
    \begin{aligned}
      \overline{X}_1 &= X_1^2X_2X_4^2X_5 \\
      \overline{X}_2 &= X_1X_2X_3\\
      S_{X1} & = X_1X_2X_3X_4X_5X_6
    \end{aligned}
  $}}\quad
\tikzfigscale{0.8}{Z-X_enc}\quad=\quad \tikzfigscale{0.8}{X-Z_enc_simplified}\quad \vcenter{
  \hbox{$
    \begin{aligned}
      \overline{Z}_1 &= Z_4Z_5^2\\
      \overline{Z}_2 &= Z_3Z_6^2\\
      S_{Z1} & = Z_1Z_2Z_3\\
      S_{Z2} & = Z_2Z_3^2Z_4^2Z_5\;\;\;\;\;\;\;\;\;\;\;\,\\
      S_{Z3} & = Z_4Z_5Z_6
    \end{aligned}
  $}
}
\end{equation}
Recall (see Eq.~\ref{eq:symmetricANDzx}) that the overall symmetric T-depth 1 circuit for AND gate consists of outer Clifford gates in addition to the inner CX circuit. The encoders above are for the CSS code corresponding to this inner CX circuit.

By pushing logical X and Z spiders through the full encoder, we obtain the following logical operators for the ``overall'' code implementing transversal AND:
\begin{equation}
\begin{array}{@{}l@{\hspace{1.5cm}}l@{}}
\overline{X}_1 = X_1^2 X_2 X_4^2 X_5
& \overline{Z}_1 = X_1 X_2 X_3 Z_4 Z_5^2 \\[4pt]
\overline{X}_2 = X_1 X_2^2 Z_3^2 X_4 X_5^2 Z_6
& \overline{Z}_2 = X_1 X_2 X_3
\end{array}
\end{equation}
We verify that the code distance is 2 by checking the minimum weight of all possible logical operators (such as by multiplying each logical operator with all possible linear combinations of the stabilizers, for this small code).

We will use $E_{pf}$ to refer to the inner CSS encoder in Eq.~\ref{eq:Epf} (comprised of phase-free ``pf'' spiders) for the CSS code implementing the $\ket{0}$-controlled Z. We will use $E_{AND}$ to refer to the full encoder (consisting of outer Clifford gates and $E_{pf}$) for the $\llbracket 6,2,2\rrbracket$ qutrit code implementing transversal AND.

\subsection{Getting to Distance 4 Through Code Concatenation}\label{sec:concat}
We remark that our $\llbracket 6,2,2 \rrbracket$ qutrit code with a transversal AND gate (implemented by 3 T and 3 T$^\dagger$ gates), can be concatenated as the outer code with the inner $\llbracket 8,1,2 \rrbracket$ qutrit CSS code with transversal T gate from~\cite{Campbell2012msdprime}.
As both are qutrit CSS codes and the inner CSS code's X and Z logical operators are implemented by physical X's and Z's respectively, the resulting code distance is the product of that of the inner and outer codes\footnote{This is a well-known result \cite{gottesman2024survivingch9}, deducible by presuming that any logical error can occur if at least $d_{\text{outer}}$ logical qubits of the inner CSS code experienced errors, each of which had to have experienced at least $d_{\text{inner}}$ physical errors.}: $\llbracket 48,2,4 \rrbracket$, with transversal AND gate implemented by 24 T and 24 T$^\dagger$ gates.

\section{Qubit Subspace Codes}
In this section, we will introduce Qubit Subspace Projection as a logical operation in qutrit codes. In addition, depending on how these results are utilized, these can be viewed from the perspectives of gauge fixing of subsystem codes, and as adding additional stabilizers to construct mixed qubit-qutrit codes which we call Qubit Subspace Codes.
We begin by deriving for the qutrit setting, the physical action of logical ZX diagrams for CSS codes in Section~\ref{sec:pte} of the appendix.

\subsection{Logical Qubit, Physical Qutrits}
We describe how to derive a quantum code that encodes a logical qubit ($k=1$) using $n=6$ physical qutrits. While we suspect this first Qubit Subspace Code is of more limited use than the one in the next subsection, this also serves as a simpler demonstration of code construction techniques:

Begin with $E_{pf}$, the encoder for the $\llbracket 6,2,2\rrbracket$ inner qutrit CSS code implementing $\ket{0}$-controlled Z (henceforth denoted ZCZ for short). Fix the bottom logical qutrit to $\ket{0}_L$, apply logical $X$ to the top qutrit (to the left of the encoder), measure the bottom logical qutrit in the X basis, and check if it is $\bra{+}_L$ using Steane error correction \cite{steane1997active} on a $\ket{0}_L\ket{+}_L$ resource state. This procedure, presented in more detail in Section~\ref{sec:andmsdi} of the appendix, uses gauge fixing to enable a logical operation, by projecting the top logical qutrit into the qubit subspace.

\subsection{Qubit Subspace Projection}
Now, we will show how to obtain a code with $k=2$ logical qubits and $n=6$ physical qutrits that has the same stabilizers as that of $E_{AND}$, but has two additional stabilizers which project the two logical qutrits into the qubit subspace.
The idea is that on the binary-valued inputs and output of the AND gate, physical projective measurements can be made that ensure that the logical qutrit states are in their qubit subspace:
\begin{equation}\label{eq:qspboxeq}
    \tikzfigscale{0.8}{subspace_proj_AND}
\end{equation}

The following logical circuit implements $\Pi_L$, the projection of a qutrit into its $\ket{0}$ and $\ket{1}$ subspace. We call this the qubit subspace projection (SP) gadget. Using the ZH diagram for $\lnot\ket{0}$-controlled $X$ from \cite{Roy_2023}, we can derive the ZX diagram for the SP gadget, as shown on the right hand side.
\begin{equation}
    \tikzfigscale{0.8}{leakage_gadget} = \tikzfigscale{0.8}{leakage_gadget_2} = \tikzfigscale{0.8}{leakage_gadget_zx} = \tikzfigscale{0.8}{leakage_gadget_zx_simplified}
\end{equation}
To find the physical implementation of the SP gadget in the $\llbracket 6,2,2\rrbracket$ code, we can push this diagram through the full Clifford encoder.
Below is the derivation for obtaining the physical SP gadget for the second (bottom) logical qutrit\footnote{The second logical qutrit is the ``answer'' qutrit of the AND circuit, mapping to $a\land b$. The physical SP gadget for the first (top) logical qutrit can be analogously derived by pushing the logical SP gadget as input to the top wire through the encoder.}. This is done by plugging in the SP gadget into the second logical input wire, then applying rewrite rules to push this through the outer Cliffords.
\begin{align}
    & \tikzfigscale{0.8}{PTE_full_enc_1_simplified}  =  \tikzfigscale{0.8}{PTE_full_enc_3_simplified}
\end{align}
After using Lemma~\ref{lemma:pte} to push the gadget through the inner CSS encoder $E_{pf}$, we find the physical implementation of the SP gadget for the second logical qutrit, supported on physical qutrits 3 and 6 as it matches $\overline{Z_2} = Z_3 Z_6^2$ of $E_{pf}$:
\begin{equation}\label{eq:qsp}
\tikzfigscale{0.8}{Pi_bot} \;\;=\;\; \tikzfigscale{0.8}{physicalQSPGflippedallwires} \;\;=\;\;\tikzfigscale{0.8}{physicalQSPGcirc}
\end{equation}
The right hand side is an equivalent diagram in which the gadget is viewed as a resource state.
We observe that the left hand side of the gadget, starting from the $\ket{1}$ input (represented by the red $\frac{2}{1}$ spider) to before the second $H^\dagger$, is a qutrit state prepared as the \textit{qubit} $\ket{+}$, an equal superposition of $\ket{0}$ and $\ket{1}$.

Beyond a static code construction, this technique can be viewed as a form of code switching through subsystem code gauge fixing. This is because we derive the projection operator that uses the qubit $\ket{+}$ state as part of the resource state. When the gauge logical qutrit state is this full gadget state (preparing a qubit $\ket{+}$ followed by $H^\dagger$), it results in a logical qubit. When instead the gauge logical qutrit state the qutrit $\ket{+}$ state, there is no Qubit Subspace Projection.

\section{Magic State Distillation and Injection}
The $\ket{0}$-controlled Z gate is deterministically injectable because it is a diagonal gate in the third level of the Clifford hierarchy, which we prove in Sections~\ref{sec:cczmsd} to~\ref{sec:andmsdi} of the appendix:
\begin{equation}
    \input{ZCZ/ZCZinj.tikz}
\end{equation}
There we also prove correctness of magic state distillation and deterministic injection of the AND gate for qutrit CSS codes. The injection circuit is:
\begin{equation}\label{eq:msd}
    \tikzfigscale{1.0}{ANDmsdcorrectedRIGHT}
\end{equation}
The correction is a Clifford unitary where all but one gate---the $Z(2b,2b)$ gate---are Pauli or CX gates, and hence transversally implementable on CSS codes~\cite{GottesmanD1997thesis}.

\section{Conclusion}
In this work, we demonstrated that moving to qutrit systems enables an efficient unitary realization of the classical binary AND gate within the Clifford+T framework, and we developed explicit low–T-count constructions such as generalizing to multi-input AND. By reinterpreting symmetric compute–phase–uncompute decompositions as encoding and decoding circuits, we showed how to design qutrit stabilizer codes that admit transversal non-Clifford implementations of logical AND, including a $\llbracket 6,2,2\rrbracket$ code and a concatenated extension with increased distance. We further introduced mixed-dimensional subspace constructions and protocols for magic state distillation and deterministic injection compatible with qutrit CSS codes. Together, these results illustrate how higher-dimensional and mixed-dimensional code design can expand the landscape of fault-tolerant logical gates beyond qubit-based constraints, and motivate further investigation of their practical performance and scalability.

A natural question to investigate through in-depth comparative study would be how theorized qutrit codes perform under different noise models and hardware constraints. These could utilize recent software libraries for qudit computation such as~\cite{mato2024mqtqudits, kabir2025sdim}, in order to benchmark via numerical simulation against qubit codes with transversal CCZ gates, such as the $\llbracket 8,3,2 \rrbracket$ color code (implemented by physical T and T$^\dagger$ gates) or the $\llbracket 27,3,3\rrbracket$ code of~\cite{li2025jump} (implemented by a depth-2 physical CCZ circuit).

The other promising direction is to continue pushing the theory of higher-dimensional QEC code constructions. For discovering more codes and dynamic QEC protocols for a fixed qudit dimension, we can leverage normal forms such as for qutrit Clifford circuits~\cite{Li_2025}, for qubit Clifford isometries~\cite{khesin2025graph, Yeh2026zxstabtab} (which we expect admit straightforward generalization to prime-dimensional qudits due to~\cite{Roy_2023,PoorB2023qupitzxtra}), and qutrit Clifford+T circuits~\cite{glaudell_exact_2024}.

\begin{credits}
\subsubsection{\ackname}
Thank you to Nadish de Silva for explaining the approach for efficient multi-qudit gate teleportation described in~\cite{deSilvaN2021quditgatetp}.
We would like to thank Noah Goss, Matt McEwen, Costin Iancu, and Ed Younis for discussions about useful circuit constructions for three-level quantum systems.
We would like to thank Sarah Meng Li and Jacky (Jiaxin) Huang for the collaboration on pushing through the encoder and gauge fixing in~\cite{HuangJ2023zxcss} that greatly informed this work.
We would like to thank Aleks Kissinger for discussions on the transversal CCZ of $\llbracket 8,3,2\rrbracket$ code.
LY would like to thank the Google PhD Fellowship, and the Tencent Post-Doctoral Research Fellowship at the University of Cambridge, for funding.

\end{credits}

\bibliographystyle{splncs04}
\bibliography{refs.bib}
\newpage
\appendix
\include{appendix.tex}
\end{document}

%% file: zx.tikzstyles
\tikzstyle{box}=[shape=rectangle, text height=1.5ex, text depth=0.25ex, yshift=0.5mm, fill=white, draw=black, minimum height=5mm, yshift=-0.5mm, minimum width=5mm, font={\small}]
\tikzstyle{Z dot}=[inner sep=0mm, minimum size=2mm, shape=circle, draw=black, fill={rgb,255: red,216; green,248; blue,217}, tikzit fill={rgb,255: red,216; green,248; blue,217}]
\tikzstyle{Z phase dot}=[minimum size=4.75mm, font={\footnotesize}, shape=rectangle, rounded corners=1.9mm, inner sep=0.1mm, outer sep=-2mm, scale=0.8, tikzit shape=circle, draw=black, fill={rgb,255: red,216; green,248; blue,216}, tikzit draw=blue, tikzit fill={rgb,255: red,216; green,248; blue,216}]
\tikzstyle{X dot}=[Z dot, shape=circle, draw=black, fill={rgb,255: red,232; green,165; blue,164}, tikzit fill={rgb,255: red,232; green,165; blue,164}]
\tikzstyle{X phase dot}=[Z phase dot, tikzit shape=circle, tikzit fill={rgb,255: red,232; green,165; blue,164}, fill={rgb,255: red,232; green,165; blue,164}, font={\footnotesize}, tikzit draw=black]
\tikzstyle{hadamard}=[fill={rgb,255: red,255; green,255; blue,130}, draw=black, shape=rectangle, inner sep=0.6mm, minimum height=1.5mm, minimum width=1.5mm]
\tikzstyle{vertex}=[inner sep=0mm, minimum size=1mm, shape=circle, draw=black, fill=black]
\tikzstyle{vertex set}=[inner sep=0mm, minimum size=1mm, shape=circle, draw=black, fill=white, font={\footnotesize\boldmath}]
\tikzstyle{white dot}=[Z dot]
\tikzstyle{gray dot}=[X dot]
\tikzstyle{white phase dot}=[Z phase dot]
\tikzstyle{gray phase dot}=[X phase dot]
\tikzstyle{tiny none}=[none, font={\tiny}]
\tikzstyle{filament}=[hadamard, fill=yellow, draw=none, minimum height=0.01mm]
\tikzstyle{unit circle}=[shape=circle, minimum size=42.5 mm, fill=none, draw={rgb,255: red,223; green,223; blue,223}, tikzit draw={rgb,255: red,223; green,223; blue,223}]
\tikzstyle{white circle}=[shape=ellipse, fill=white, draw=black]
\tikzstyle{2gate}=[fill=white, draw=black, shape=rectangle, minimum width=0.75cm, minimum height=1.25cm]
\tikzstyle{small dot}=[vertex, minimum size=1.5 mm, draw=black, tikzit draw=black, tikzit fill=black, tikzit shape=circle]
\tikzstyle{targ}=[vertex set, minimum size=0.5mm, inner sep=-0.5mm, tikzit shape=circle, shape=circle, tikzit draw=black]
\tikzstyle{2 control}=[vertex set, draw=black, inner sep=0.5pt]
\tikzstyle{1 control}=[2 control, draw=black]
\tikzstyle{0 control}=[2 control, draw=black]
\tikzstyle{neg0 control}=[0 control, inner sep=0 mm, fill=white, draw=black]
\tikzstyle{gn}=[inner sep=0mm, minimum size=2mm, shape=circle, draw=black, fill={rgb,255: red,216; green,248; blue,216}, tikzit fill={rgb,255: red,216; green,248; blue,216}]
\tikzstyle{rn}=[gn, shape=circle, draw=black, fill={rgb,255: red,232; green,165; blue,165}, tikzit fill={rgb,255: red,232; green,165; blue,165}]
\tikzstyle{gn_phase}=[shape=rectangle, fill={zx_green}, draw=black, minimum size=1.2em, rounded corners=0.48em, inner sep=0.2em, outer sep=-0.2em, scale=0.8, font={\footnotesize}, tikzit shape=circle, tikzit fill={rgb,255: red,181; green,215; blue,181}]
\tikzstyle{rn_phase}=[{gn_phase}, fill={rgb,255: red,232; green,165; blue,165}, draw=black, tikzit fill={rgb,255: red,232; green,165; blue,165}]
\tikzstyle{small hadamard}=[fill=yellow, draw, inner sep=0.6mm, minimum height=1.5mm, minimum width=1.5mm, tikzit shape=rectangle]
\tikzstyle{dualizer}=[hadamard, fill=white, shape=rectangle, tikzit fill=white, tikzit draw=black, draw=black]

\tikzstyle{directedarrow}=[draw={rgb,255: red,223; green,223; blue,223}, ->, tikzit draw={rgb,255: red,223; green,223; blue,223}, line width=1 pt]
\tikzstyle{simple}=[-]
\tikzstyle{hadamard edge}=[-, color={rgb,255: red,100; green,200; blue,248}, dashed, dash pattern=on 2pt off 0.7pt, tikzit draw={rgb,255: red,120; green,220; blue,248}, line width=0.2mm]
\tikzstyle{brace edge}=[-, tikzit draw=blue, decorate, decoration={brace,amplitude=1mm,raise=-1mm}]
\tikzstyle{gray}=[-, draw={rgb,255: red,223; green,223; blue,223}, line width=1 pt]
\tikzstyle{arrow}=[<-, draw={rgb,255: red,128; green,128; blue,128}]
\tikzstyle{double-arrow}=[draw={rgb,255: red,128; green,128; blue,128}, <->, line width=1.5 pt]
\tikzstyle{dashed edge}=[-, dashed, dash pattern=on 2pt off 0.5pt, draw=black]
\tikzstyle{diredge}=[->]
\tikzstyle{double edge}=[-, double, shorten <=-1mm, shorten >=-1mm, double distance=2pt]
\tikzstyle{thin}=[-, line width=0.05mm]
\tikzstyle{thin gray}=[-, draw={rgb,255: red,223; green,223; blue,223}, line width=0.05mm]
\tikzstyle{less thin}=[-, line width=0.1mm]
\tikzstyle{dashed gray edge}=[-, dashed edge, draw={rgb,255: red,223; green,223; blue,223}]
\tikzstyle{light right directed arrow}=[->, directedarrow, draw={rgb,255: red,223; green,223; blue,223}, line width=0.2mm]
\tikzstyle{diredge0.3}=[->, line width=0.3 mm]
\tikzstyle{less thin gray}=[-, draw={rgb,255: red,223; green,223; blue,223}]
\tikzstyle{hadamard arrows}=[dashed, color={rgb,255: red,128; green,0; blue,128}, tikzit draw={rgb,255: red,128; green,0; blue,128}, <->, dash pattern=on 2pt off 0.7pt]
\tikzstyle{double arrow}=[<->]
\tikzstyle{light blue edge}=[-, color={rgb,255: red,100; green,200; blue,248}, tikzit draw={rgb,255: red,100; green,200; blue,248}]
\tikzstyle{black0.3}=[-, line width=0.3mm]
\tikzstyle{pink0.3}=[-, line width=0.3mm, color={rgb,255: red,255; green,100; blue,200}, tikzit draw={rgb,255: red,255; green,100; blue,200}]
\tikzstyle{light blue0.3}=[-, line width=0.3mm, color={rgb,255: red,100; green,200; blue,248}, tikzit draw={rgb,255: red,100; green,200; blue,248}]
\tikzstyle{green0.3}=[-, color={rgb,255: red,100; green,255; blue,50}, tikzit draw={rgb,255: red,100; green,255; blue,50}, line width=0.3mm]
\tikzstyle{diredge light0.3}=[->, diredge, color=blue, line width=0.3mm, tikzit draw=blue]
\tikzstyle{light blue diredge0.08}=[light blue edge, ->, tikzit draw={rgb,255: red,100; green,200; blue,248}, line width=0.08mm]
\tikzstyle{filllayer}=[-, fill={rgb,255: red,250; green,225; blue,200}, opacity=0.5]
\tikzstyle{snekedgeb}=[-, snekedge, line width=1 pt, draw={rgb,255: red,38; green,100; blue,100}]
\tikzstyle{filllayergray}=[-, filllayer, fill={rgb,255: red,150; green,150; blue,150}, draw=white, line width=0 pt]
\tikzstyle{filllayerlight}=[-, filllayergray, fill={rgb,255: red,220; green,220; blue,220}]
\tikzstyle{h edge}=[-, dashed, dash pattern=on 3pt off 2.5pt, draw={rgb,255: red,42; green,134; blue,255}, line width=0.9 pt]
\tikzstyle{filllayerwhite}=[-, filllayer, fill=white, opacity=1]
\tikzstyle{white edge}=[-, draw=white]
\tikzstyle{filllayerlightlightgray}=[-, filllayergray, fill={rgb,255: red,240; green,240; blue,240}]
\tikzstyle{geo}=[-, double, dash pattern=on 1pt off 0.7pt, line width=0.2pt, draw={rgb,255: red,209; green,129; blue,223}]

%% file: qcircuit_common.tex
\usepackage[braket, qm]{qcircuit}
\usepackage{amsmath}
\usepackage{tikz}

\newcommand*\qcontrolcolor[2]{\push{\footnotesize \tikz[baseline=(char.base)]{
                              \node[shape=circle,draw,inner sep=0.6pt,color=#1] (char) {#2};}}
                              \qw}  

\newcommand*\onecontrol{\qcontrolcolor{black}{1}}
\newcommand*\twocontrol{\qcontrolcolor{black}{2}}

%% file: figures/qutrit_Z_spider.tikz
\begin{tikzpicture}
	\begin{pgfonlayer}{nodelayer}
		\node [style=Z phase dot] (0) at (0, 0) {$\frac{\ \alpha\ }{\ \beta\ }$};
		\node [style=none] (1) at (-2, 1) {};
		\node [style=none] (2) at (-2, -1) {};
		\node [style=none] (3) at (2, 1) {};
		\node [style=none] (4) at (2, -1) {};
		\node [style=none] (7) at (-1.5, 0.25) {$\vdots$};
		\node [style=none] (8) at (1.5, 0.25) {$\vdots$};
	\end{pgfonlayer}
	\begin{pgfonlayer}{edgelayer}
		\draw [bend left] (1.center) to (0);
		\draw [bend right] (2.center) to (0);
		\draw [bend left] (0) to (3.center);
		\draw [bend right] (0) to (4.center);
	\end{pgfonlayer}
\end{tikzpicture}

%% file: figures/qutrit_X_spider.tikz
\begin{tikzpicture}
	\begin{pgfonlayer}{nodelayer}
		\node [style=none] (1) at (-2, 1) {};
		\node [style=none] (2) at (-2, -1) {};
		\node [style=none] (3) at (2, 1) {};
		\node [style=none] (4) at (2, -1) {};
		\node [style=none] (7) at (-1.5, 0.25) {$\vdots$};
		\node [style=none] (8) at (1.5, 0.25) {$\vdots$};
		\node [style=X phase dot] (9) at (0, 0) {$\frac{\ \alpha\ }{\ \beta\ }$};
	\end{pgfonlayer}
	\begin{pgfonlayer}{edgelayer}
		\draw [bend left] (1.center) to (9);
		\draw [bend right] (2.center) to (9);
		\draw [bend left] (9) to (3.center);
		\draw [bend left] (4.center) to (9);
	\end{pgfonlayer}
\end{tikzpicture}

%% file: figures/qutrit_H.tikz
\begin{tikzpicture}
	\begin{pgfonlayer}{nodelayer}
		\node [style=none] (0) at (-1.5, 0) {};
		\node [style=none] (1) at (1.5, 0) {};
		\node [style=hadamard] (2) at (0, 0) {1};
	\end{pgfonlayer}
	\begin{pgfonlayer}{edgelayer}
		\draw [in=180, out=0] (0.center) to (2);
		\draw (2) to (1.center);
	\end{pgfonlayer}
\end{tikzpicture}

%% file: figures/qutrit_swap.tikz
\begin{tikzpicture}
	\begin{pgfonlayer}{nodelayer}
		\node [style=none] (0) at (-2, 1) {};
		\node [style=none] (1) at (-2, -1) {};
		\node [style=none] (2) at (2, 1) {};
		\node [style=none] (3) at (2, -1) {};
		\node [style=none] (4) at (0, 0) {};
	\end{pgfonlayer}
	\begin{pgfonlayer}{edgelayer}
		\draw [bend left] (0.center) to (4.center);
		\draw [bend right] (4.center) to (3.center);
		\draw [bend left] (4.center) to (2.center);
		\draw [bend left] (4.center) to (1.center);
	\end{pgfonlayer}
\end{tikzpicture}

%% file: figures/qutrit_id.tikz
\begin{tikzpicture}
	\begin{pgfonlayer}{nodelayer}
		\node [style=none] (0) at (-1, 0) {};
		\node [style=none] (1) at (1, 0) {};
	\end{pgfonlayer}
	\begin{pgfonlayer}{edgelayer}
		\draw (0.center) to (1.center);
	\end{pgfonlayer}
\end{tikzpicture}

%% file: figures/ZCZ/ZCZinj.tikz
\begin{tikzpicture}
	\begin{pgfonlayer}{nodelayer}
		\node [style=none] (4) at (-10.625, -1.25) {$|\psi_{\textnormal{\scriptsize ZCZ}} \rangle$};
		\node [style=none] (7) at (-1.75, 1.25) {$=$};
		\node [style=box] (54) at (-6.75, -0.625) {};
		\node [style=box] (55) at (-6.75, -0.625) {$\!X\!$};
		\node [style=vertex set] (56) at (-6.75, 1.875) {$\Lambda$};
		\node [style=none] (57) at (-11.75, 1.875) {};
		\node [style=none] (58) at (-2.5, 1.875) {};
		\node [style=none] (59) at (-4.5, -0.625) {};
		\node [style=none] (60) at (-9.25, -0.625) {};
		\node [style=box] (61) at (-8.25, -0.625) {$\!H^2\!$};
		\node [style=box] (62) at (-5.5, -1.875) {};
		\node [style=box] (63) at (-5.5, -1.875) {$\!X\!$};
		\node [style=vertex set] (64) at (-5.5, 0.625) {$\Lambda$};
		\node [style=none] (65) at (-11.75, 0.625) {};
		\node [style=none] (66) at (-2.5, 0.625) {};
		\node [style=none] (68) at (-9.25, -1.875) {};
		\node [style=box] (69) at (-8.25, -1.875) {$\!H^2\!$};
		\node [style=none] (70) at (-9.375, -0.625) {};
		\node [style=none] (71) at (-9.375, -1.875) {};
		\node [style=none] (72) at (-9.5, -1.25) {};
		\node [style=none] (84) at (-4.45, -0.875) {};
		\node [style=none] (85) at (-4.5, -0.275) {};
		\node [style=none] (86) at (-4.5, -1.025) {};
		\node [style=none] (87) at (-3.5, -1.025) {};
		\node [style=none] (88) at (-3.5, -0.275) {};
		\node [style=none] (89) at (-3.55, -0.875) {};
		\node [style=none] (90) at (-4, -0.95) {};
		\node [style=none] (91) at (-3.75, -0.325) {};
		\node [style=none] (93) at (-4.275, -0.275) {};
		\node [style=none] (94) at (-4.225, -0.275) {};
		\node [style=none] (95) at (-4.5, -1.875) {};
		\node [style=none] (96) at (-4.45, -2.125) {};
		\node [style=none] (97) at (-4.5, -1.525) {};
		\node [style=none] (98) at (-4.5, -2.275) {};
		\node [style=none] (99) at (-3.5, -2.275) {};
		\node [style=none] (100) at (-3.5, -1.525) {};
		\node [style=none] (101) at (-3.55, -2.125) {};
		\node [style=none] (102) at (-4, -2.2) {};
		\node [style=none] (103) at (-3.75, -1.575) {};
		\node [style=none] (104) at (-4.275, -1.525) {};
		\node [style=none] (105) at (-4.225, -1.525) {};
		\node [style=none] (159) at (9.125, 1.875) {};
		\node [style=none] (165) at (9.125, 0.625) {};
		\node [style=none] (171) at (-4.5, 2) {};
		\node [style=none] (172) at (-4.5, 0.5) {};
		\node [style=none] (173) at (-3.5, 0.5) {};
		\node [style=none] (174) at (-3.5, 2) {};
		\node [style=none] (175) at (-3.5, 1.875) {};
		\node [style=none] (176) at (-3.5, 0.625) {};
		\node [style=none] (177) at (-4.5, 1.875) {};
		\node [style=none] (178) at (-4.5, 0.625) {};
		\node [style=none] (179) at (-4, 1.25) {$C$};
		\node [style=none] (180) at (-4.05, 0.5) {};
		\node [style=none] (181) at (-3.95, 0.5) {};
		\node [style=none] (182) at (-4.05, -0.275) {};
		\node [style=none] (183) at (-3.95, -0.275) {};
		\node [style=none] (184) at (-4.05, -1.025) {};
		\node [style=none] (185) at (-3.95, -1.025) {};
		\node [style=none] (186) at (-4.05, -1.525) {};
		\node [style=none] (187) at (-3.95, -1.525) {};
		\node [style=Z phase dot] (193) at (7.375, 0.65) {$\frac{a\!^2}{2a\!^2}$};
		\node [style=Z phase dot] (194) at (7.375, 1.875) {$\frac{2ab}{ab}$};
		\node [style=Z dot] (195) at (8.5, 0.65) {};
		\node [style=hadamard] (196) at (8.5, 1.175) {\scriptsize $a$};
		\node [style=Z phase dot] (197) at (8.5, 1.875) {$\frac{2b}{2b}$};
		\node [style=box] (276) at (0, 0.625) {$\!X\!$};
		\node [style=vertex set] (277) at (0, 1.875) {0};
		\node [style=none] (282) at (-0.875, 1.875) {};
		\node [style=none] (283) at (0.875, 1.875) {};
		\node [style=none] (284) at (-0.875, 0.625) {};
		\node [style=none] (285) at (0.875, 0.625) {};
		\node [style=none] (286) at (6.625, 1.875) {};
		\node [style=none] (287) at (6.625, 0.625) {};
		\node [style=none] (288) at (5.375, 1.875) {};
		\node [style=none] (289) at (5.375, 0.625) {};
		\node [style=none] (290) at (4.125, 2) {};
		\node [style=none] (291) at (4.125, 0.5) {};
		\node [style=none] (292) at (5.125, 0.5) {};
		\node [style=none] (293) at (5.125, 2) {};
		\node [style=none] (294) at (5.125, 1.875) {};
		\node [style=none] (295) at (5.125, 0.625) {};
		\node [style=none] (296) at (4.125, 1.875) {};
		\node [style=none] (297) at (4.125, 0.625) {};
		\node [style=none] (298) at (4.625, 1.25) {$C$};
		\node [style=none] (299) at (4.575, 0.5) {};
		\node [style=none] (300) at (4.675, 0.5) {};
		\node [style=none] (301) at (3.875, 1.875) {};
		\node [style=none] (302) at (3.875, 0.625) {};
		\node [style=none] (303) at (6, 1.25) {$=$};
		\node [style=none] (304) at (2.25, 0.75) {, where};
		\node [style=none] (305) at (-3.225, -0.575) {};
		\node [style=none] (306) at (-3.225, -0.675) {};
		\node [style=none] (307) at (-3.5, -0.575) {};
		\node [style=none] (308) at (-3.5, -0.675) {};
		\node [style=none] (309) at (-3.225, -1.825) {};
		\node [style=none] (310) at (-3.225, -1.925) {};
		\node [style=none] (311) at (-3.5, -1.825) {};
		\node [style=none] (312) at (-3.5, -1.925) {};
		\node [style=none] (313) at (-2.85, -0.625) {$\langle a|$};
		\node [style=none] (314) at (-2.85, -1.875) {$\langle b|$};
	\end{pgfonlayer}
	\begin{pgfonlayer}{edgelayer}
		\draw [style=simple] (56) to (55);
		\draw [style=simple] (57.center) to (56);
		\draw [style=simple] (60.center) to (55);
		\draw [style=simple] (55) to (59.center);
		\draw [style=simple] (64) to (63);
		\draw [style=simple] (65.center) to (64);
		\draw [style=simple] (68.center) to (63);
		\draw [style=less thin] (70.center)
			 to [in=15, out=180, looseness=0.75] (72.center)
			 to [in=180, out=-15, looseness=0.75] (71.center);
		\draw (88.center)
			 to (85.center)
			 to (86.center)
			 to (87.center)
			 to cycle;
		\draw [bend left=75, looseness=1.50] (84.center) to (89.center);
		\draw (91.center) to (90.center);
		\draw (100.center) to (97.center);
		\draw (97.center) to (98.center);
		\draw (98.center) to (99.center);
		\draw (99.center) to (100.center);
		\draw [bend left=75, looseness=1.50] (96.center) to (101.center);
		\draw (103.center) to (102.center);
		\draw (95.center) to (63);
		\draw (172.center)
			 to (171.center)
			 to (174.center)
			 to (173.center)
			 to cycle;
		\draw (56) to (177.center);
		\draw (64) to (178.center);
		\draw (176.center) to (66.center);
		\draw (58.center) to (175.center);
		\draw (180.center) to (182.center);
		\draw (181.center) to (183.center);
		\draw (184.center) to (186.center);
		\draw (185.center) to (187.center);
		\draw (197) to (195);
		\draw [style=simple] (277) to (276);
		\draw (283.center) to (282.center);
		\draw (284.center) to (285.center);
		\draw (286.center) to (159.center);
		\draw (165.center) to (287.center);
		\draw (291.center)
			 to (290.center)
			 to (293.center)
			 to (292.center)
			 to cycle;
		\draw (295.center) to (289.center);
		\draw (288.center) to (294.center);
		\draw (301.center) to (296.center);
		\draw (302.center) to (297.center);
		\draw (305.center) to (307.center);
		\draw (306.center) to (308.center);
		\draw (309.center) to (311.center);
		\draw (310.center) to (312.center);
	\end{pgfonlayer}
\end{tikzpicture}

%% file: appendix.tex
\section{Qutrit ZX-calculus Rules}
\begin{figure}[h]
    \centering
    \includegraphics[width=0.8\linewidth]{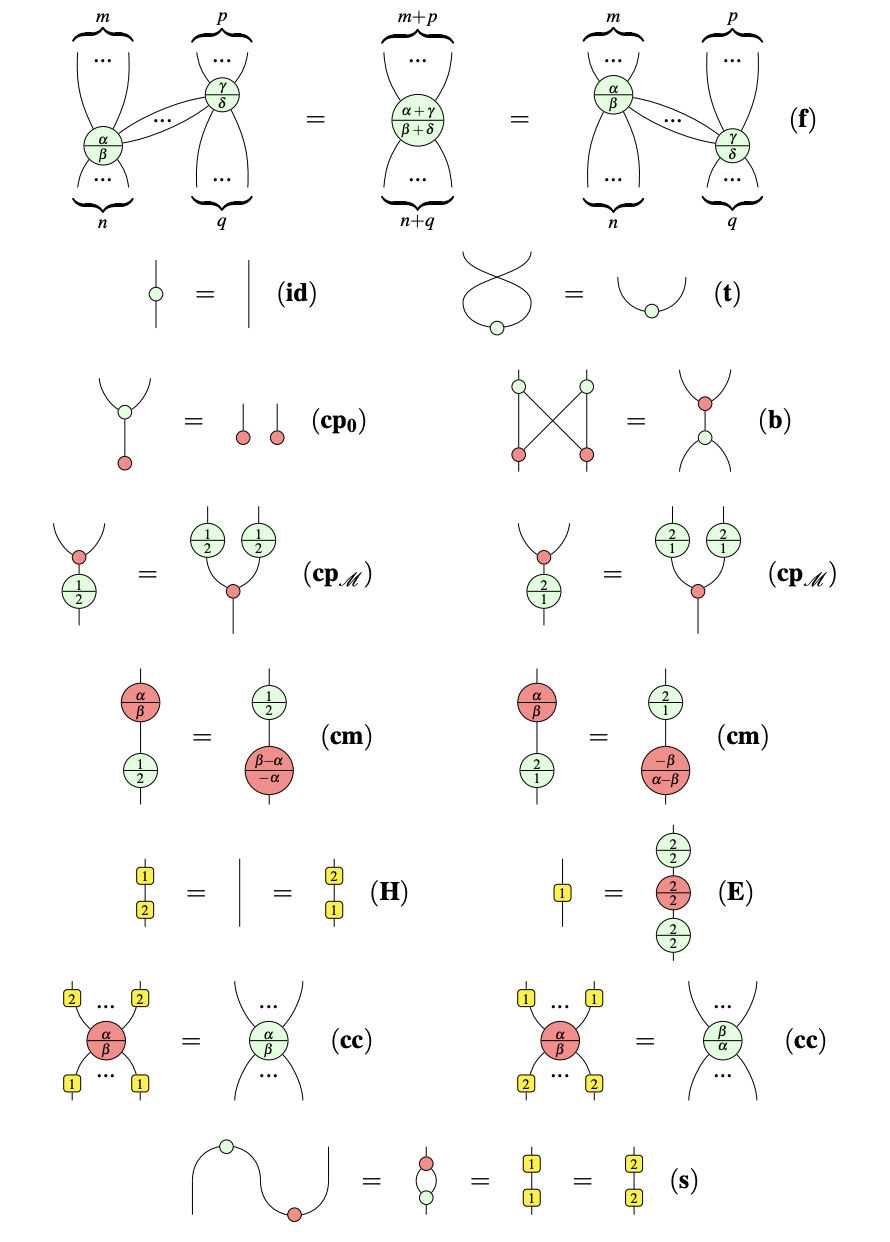}
    \caption{Rewrite rules for the qutrit ZX-calculus, reprinted from Figure 1 in~\cite{townsend-teague_simplification_2022}}
    \label{fig:townsendteaguequtritrules}
\end{figure}

\subsection{Pushing through the encoder} \label{sec:pte}
``Pushing through the encoder'' refers to the idea that, for any encoder $E$ for a CSS code (which is a phase-free diagram that is an isometry), any diagram $\mathcal{D}$ on the left of $E$ (which is the logical operator) has a corresponding diagram $\mathcal{D}_P$ on the right of $E$ (which is a physical implementation of $\mathcal{D}$). This was developed for qubit CSS codes in \cite{HuangJ2023zxcss}. We generalize this idea to qutrits as follows:
\begin{lemma}[``Pushing through the encoder'' (PTE) for qutrits \label{lemma:pte}]
    Let $E$ be the encoder of a qutrit CSS code. For any ZX diagram $\mathcal{D}_L$ on the left of $E$, it has a corresponding diagram $\mathcal{D}_P$ on the right of $E$ such that $E\mathcal{D}_L = \mathcal{D}_PE$.
\end{lemma}
\begin{proof}
Similar to the proof for Proposition 3.1 of \cite{HuangJ2023zxcss}, any Z or X spiders in $\mathcal{D}_L$ on the logical wires that has phases or multiple legs can be unfused so that each logical Z or X spider on the logical qutrit wire has just one external wire. Using strong complementarity with the corresponding choice of diagram of the encoder (for a Z spider, use X-Z normal form, and for X spider, use Z-normal form), one can ``push'' this supporting leg of the diagram to the right side. 
For any diagram $\mathcal{D}$ supported on an arbitrary number of logical wires,
    \begin{equation}
    \tikzfigscale{0.7}{PTEgeneral1} = \tikzfigscale{0.7}{PTEgeneral2}
\end{equation}
Note that for qutrits, there may be double wires (which can be represented as a single wire with a ``D'' dualizer box) in encoders connecting Z and X spiders. The connectivity degree (single or double wires) for each supporting logical qubit to $\mathcal{D}$ is preserved for the corresponding X spider for that logical on the right.
\end{proof}

\section{Deconstruction of the Qubit \texorpdfstring{$\llbracket 8,3,2 \rrbracket$}{[[8,3,2]]} Code}
\label{sec:interesting}

\subsection{Circuit Synthesis and ZX-calculus View}\label{sec:circzx}
Here, we will demonstrate the well-known example of the transversal CCZ gate of the $\llbracket 8,3,2 \rrbracket$ code~\cite{campbell_smallest_2016} (also called the smallest interesting color code, an instance of transversal CCZ gates described in~\cite{kubica2015unfolding}), in the ZX-calculus.

\begin{figure}[H]
    \centering
    \includegraphics[width=0.95\textwidth]{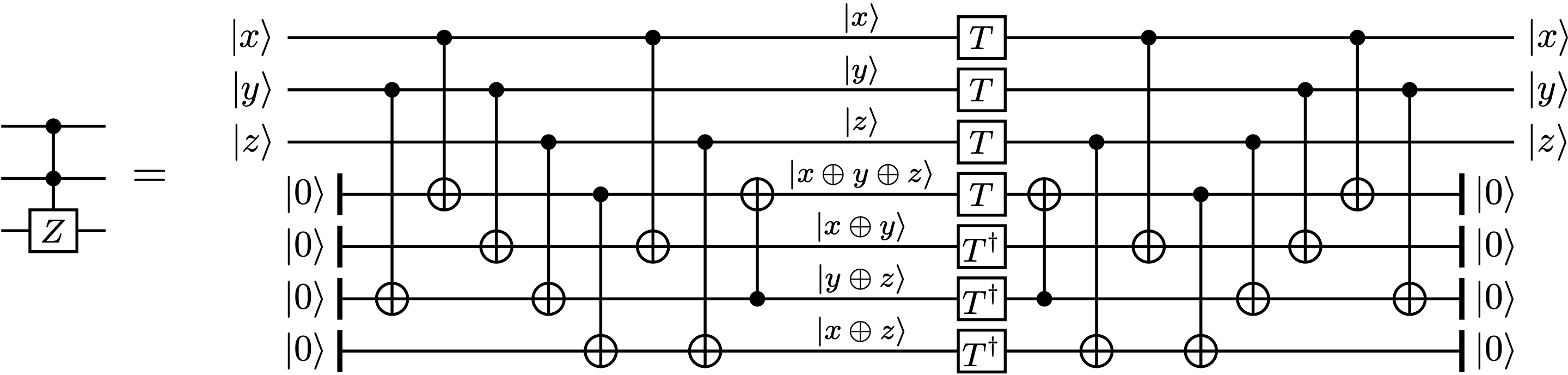}
    \caption{Figure for a T-depth one decomposition of the qubit CCZ gate, reprinted from~\cite[Figure 1]{SelingerP2013tdepthone}.}
\end{figure}
In an observation credited to Campbell~\cite{campbell_smallest_2016}, the above circuit consists of a CNOT circuit unitary embedding $E$, a layer of $T$ and $T^\dagger$ gates, and $E^\dagger$.  This therefore defines a $\llbracket 7,3,1\rrbracket$ code. Being distance 1, it has no error-correcting capabilities.
To improve the code distance, it suffices to add to the encoder $E$ (whose ZX diagram first appeared in~\cite{Garvie_2018}) one physical qubit, and a weight 8 X-type stabilizer generator connected to all 8 physical qubits:
\begin{equation}
    \input{832cubeMODIFIED.tikz}
\end{equation}

\subsection{Verifying Transversal CCZ Preservation}\label{sec:transversalccz}
This construction preserves the transversal CCZ gate, resulting in the $\llbracket 8,3,2 \rrbracket$ code. To show this, we will use the following fact about the phase polynomial representation of the CCZ gate as the map that acts as $\ket{x,y,z} \mapsto (-1)^{x\cdot y\cdot z} \ket{x,y,z}$, where $x,y,z$ are bits denoting computational basis states.
As was shown for qutrits in~\cite{wetering_building_2023}, the qubit CCZ gate is identified by the exponent of $(-1)$ being a multiplicative expression $x\cdot y\cdot z$; this can be replaced by an equivalent expression summing over modular addition.
For this reason, the qubit CCZ gate is equal to the following spider nest of phase gadgets:
\tikzstyle{tikzfig}=[baseline=-0.5em, scale=0.75]
\begin{equation}
    \input{ccz-nest.tikz}
\end{equation}
If instead all phases here are $\pm \frac{\pi}{2}$, then instead of the CCZ gate, this is the identity map on three qubits.
More generally, for $n$ qubits, having one of all possible connectivity phase gadgets of 0-$n$ legs, of phases $\pm \frac{\pi}{2^{n-2}}$ where the sign matches the parity of the number of legs, is equal to the $n$-qubit identity map.

We can substitute the encoder $E$ of the $\llbracket 8,3,2 \rrbracket$ code from~\cite{Garvie_2018}, to verify that $L = E; P; E^\dagger$ where $L$ is the CCZ gate, and $P = T T^\dagger T^\dagger T T^\dagger T T T^\dagger$.
\begin{equation}
    \tikzfigscale{0.6}{ccz-nest-2}
\end{equation}
Continuing from the last step, we now utilize the above observation to recognize that what we have is the spider nest of all 0-4 legged $\mp \frac{\pi}{4}$ phase gadgets, conjugated by all 0-3 legged $\pm \frac{\pi}{4}$ phase gadgets on the upper three qubits.  Therefore, we can then apply all 0-4 legged $\mp \frac{\pi}{4}$ phase gadgets, which equals the identity.  The end result is precisely the CCZ gate.
\begin{equation}
    \tikzfigscale{0.8}{ccz-nest-3}
\end{equation}

\input{msd.tex}

\input{qubitemulation.tex}

%% file: figures/832cubeMODIFIED.tikz
\begin{tikzpicture}
	\begin{pgfonlayer}{nodelayer}
		\node [style=none] (63) at (-1.5, 0) {$=$};
		\node [style=gn] (65) at (-3, 9) {};
		\node [style=gn] (66) at (-4.5, 9.5) {};
		\node [style=gn] (67) at (-6, 10) {};
		\node [style=rn] (69) at (-2.5, 5) {};
		\node [style=rn] (70) at (-2.5, 5.5) {};
		\node [style=rn] (71) at (-2.5, 6) {};
		\node [style=rn] (72) at (-2.5, 6.5) {};
		\node [style=rn] (77) at (-2.5, 7) {};
		\node [style=rn] (78) at (-2.5, 7.5) {};
		\node [style=rn] (79) at (-2.5, 8) {};
		\node [style=rn] (80) at (-2.5, 8.5) {};
		\node [style=none] (85) at (-6.5, 9) {};
		\node [style=none] (86) at (-6.5, 9.5) {};
		\node [style=none] (87) at (-6.5, 10) {};
		\node [style=gn] (120) at (11, 9) {};
		\node [style=gn] (121) at (9.5, 9.5) {};
		\node [style=gn] (122) at (8, 10) {};
		\node [style=gn] (123) at (7.75, 6.75) {};
		\node [style=rn] (124) at (11.5, 5) {};
		\node [style=rn] (125) at (11.5, 5.5) {};
		\node [style=rn] (126) at (11.5, 6) {};
		\node [style=rn] (127) at (11.5, 6.5) {};
		\node [style=none] (128) at (12, 6.5) {};
		\node [style=none] (129) at (12, 6) {};
		\node [style=none] (130) at (12, 5.5) {};
		\node [style=none] (131) at (12, 5) {};
		\node [style=rn] (132) at (11.5, 7) {};
		\node [style=rn] (133) at (11.5, 7.5) {};
		\node [style=rn] (134) at (11.5, 8) {};
		\node [style=rn] (135) at (11.5, 8.5) {};
		\node [style=none] (136) at (12, 8.5) {};
		\node [style=none] (137) at (12, 8) {};
		\node [style=none] (138) at (12, 7.5) {};
		\node [style=none] (139) at (12, 7) {};
		\node [style=none] (140) at (7.5, 9) {};
		\node [style=none] (141) at (7.5, 9.5) {};
		\node [style=none] (142) at (7.5, 10) {};
		\node [style=none] (143) at (12.25, 8.5) {\tiny $1$};
		\node [style=none] (144) at (12.25, 8) {\tiny $2$};
		\node [style=none] (145) at (12.25, 7.5) {\tiny $3$};
		\node [style=none] (146) at (12.25, 7) {\tiny $4$};
		\node [style=none] (147) at (12.25, 6.5) {\tiny $5$};
		\node [style=none] (148) at (12.25, 6) {\tiny $6$};
		\node [style=none] (149) at (12.25, 5.5) {\tiny $7$};
		\node [style=none] (150) at (12.25, 5) {\tiny $8$};
		\node [style=gn] (154) at (0.5, 6.75) {};
		\node [style=rn] (155) at (4.25, 5) {};
		\node [style=rn] (156) at (4.25, 5.5) {};
		\node [style=rn] (157) at (4.25, 6) {};
		\node [style=rn] (158) at (4.25, 6.5) {};
		\node [style=none] (159) at (4.75, 6.5) {};
		\node [style=none] (160) at (4.75, 6) {};
		\node [style=none] (161) at (4.75, 5.5) {};
		\node [style=none] (162) at (4.75, 5) {};
		\node [style=rn] (163) at (4.25, 7) {};
		\node [style=rn] (164) at (4.25, 7.5) {};
		\node [style=rn] (165) at (4.25, 8) {};
		\node [style=rn] (166) at (4.25, 8.5) {};
		\node [style=none] (167) at (4.75, 8.5) {};
		\node [style=none] (168) at (4.75, 8) {};
		\node [style=none] (169) at (4.75, 7.5) {};
		\node [style=none] (170) at (4.75, 7) {};
		\node [style=none] (174) at (5, 8.5) {\tiny $1$};
		\node [style=none] (175) at (5, 8) {\tiny $2$};
		\node [style=none] (176) at (5, 7.5) {\tiny $3$};
		\node [style=none] (177) at (5, 7) {\tiny $4$};
		\node [style=none] (178) at (5, 6.5) {\tiny $5$};
		\node [style=none] (179) at (5, 6) {\tiny $6$};
		\node [style=none] (180) at (5, 5.5) {\tiny $7$};
		\node [style=none] (181) at (5, 5) {\tiny $8$};
		\node [style=none] (182) at (6.25, 7) {$=$};
		\node [style=rn] (183) at (0.375, 0.5) {};
		\node [style=rn] (184) at (2.375, 0.5) {};
		\node [style=rn] (185) at (0.375, -1.25) {};
		\node [style=rn] (186) at (2.375, -1.25) {};
		\node [style=rn] (187) at (0.875, 1.5) {};
		\node [style=rn] (188) at (2.875, 1.5) {};
		\node [style=rn] (189) at (0.875, -0.25) {};
		\node [style=rn] (190) at (2.875, -0.25) {};
		\node [style=gn] (191) at (1.625, 0) {};
		\node [style=gn] (192) at (0.625, 2.5) {};
		\node [style=gn] (193) at (4.375, -0.25) {};
		\node [style=gn] (194) at (-0.125, -2.25) {};
		\node [style=none] (195) at (0.25, 2.875) {};
		\node [style=none] (196) at (4.875, -0.25) {};
		\node [style=none] (197) at (-0.5, -2.625) {};
		\node [style=none] (198) at (1.125, 1.65) {};
		\node [style=none] (199) at (0.625, 0.65) {};
		\node [style=none] (200) at (1.125, -0.1) {};
		\node [style=none] (201) at (0.625, -1.1) {};
		\node [style=none] (202) at (2.625, -1.1) {};
		\node [style=none] (203) at (3.125, -0.1) {};
		\node [style=none] (204) at (2.625, 0.65) {};
		\node [style=none] (205) at (3.125, 1.65) {};
		\node [style=rn] (206) at (8, 1.25) {};
		\node [style=rn] (207) at (12, 1.25) {};
		\node [style=rn] (208) at (9, 3.1) {};
		\node [style=rn] (209) at (13, 3.1) {};
		\node [style=gn] (210) at (10.5, 2.175) {};
		\node [style=rn] (211) at (8, -2.25) {};
		\node [style=rn] (212) at (12, -2.25) {};
		\node [style=rn] (213) at (9, -0.4) {};
		\node [style=rn] (214) at (13, -0.4) {};
		\node [style=gn] (215) at (12.5, 0.425) {};
		\node [style=gn] (216) at (10, -0.5) {};
		\node [style=gn] (217) at (10.5, 0.425) {};
		\node [style=none] (218) at (9.375, 3.175) {};
		\node [style=none] (219) at (13.375, 3.175) {};
		\node [style=none] (220) at (12.375, 1.325) {};
		\node [style=none] (221) at (8.375, 1.325) {};
		\node [style=none] (222) at (9.375, -0.325) {};
		\node [style=none] (223) at (13.375, -0.325) {};
		\node [style=none] (224) at (12.375, -2.175) {};
		\node [style=none] (225) at (8.375, -2.175) {};
		\node [style=none] (226) at (10.625, 1.75) {};
		\node [style=gn] (227) at (12.5, 0.425) {};
		\node [style=none] (228) at (12.625, 0) {};
		\node [style=none] (229) at (10.125, -0.925) {};
		\node [style=none] (230) at (6.125, 0) {$=$};
	\end{pgfonlayer}
	\begin{pgfonlayer}{edgelayer}
		\draw [in=150, out=-45, looseness=0.75] (65) to (80);
		\draw [in=135, out=-60, looseness=0.75] (65) to (79);
		\draw [in=135, out=-75, looseness=0.75] (65) to (78);
		\draw [in=-90, out=135, looseness=0.75] (77) to (65);
		\draw [in=165, out=-30, looseness=0.50] (66) to (80);
		\draw [in=150, out=-45, looseness=0.50] (66) to (79);
		\draw [in=150, out=-60, looseness=0.25] (66) to (72);
		\draw [in=150, out=-75, looseness=0.25] (66) to (71);
		\draw [in=180, out=-30, looseness=0.50] (67) to (80);
		\draw [in=165, out=-45, looseness=0.25] (67) to (78);
		\draw [in=165, out=-60, looseness=0.25] (67) to (72);
		\draw [in=165, out=-75, looseness=0.25] (67) to (70);
		\draw (87.center) to (67);
		\draw (86.center) to (66);
		\draw (65) to (85.center);
		\draw [in=150, out=-45, looseness=0.75] (120) to (135);
		\draw [in=135, out=-60, looseness=0.75] (120) to (134);
		\draw [in=135, out=-75, looseness=0.75] (120) to (133);
		\draw [in=-90, out=135, looseness=0.75] (132) to (120);
		\draw [in=165, out=-30, looseness=0.50] (121) to (135);
		\draw [in=150, out=-45, looseness=0.50] (121) to (134);
		\draw [in=150, out=-60, looseness=0.25] (121) to (127);
		\draw [in=150, out=-75, looseness=0.25] (121) to (126);
		\draw [in=180, out=-30, looseness=0.50] (122) to (135);
		\draw [in=165, out=-45, looseness=0.25] (122) to (133);
		\draw [in=165, out=-60, looseness=0.25] (122) to (127);
		\draw [in=165, out=-75, looseness=0.25] (122) to (125);
		\draw (142.center) to (122);
		\draw (141.center) to (121);
		\draw (120) to (140.center);
		\draw (135) to (136.center);
		\draw (137.center) to (134);
		\draw (133) to (138.center);
		\draw (139.center) to (132);
		\draw (127) to (128.center);
		\draw (129.center) to (126);
		\draw (125) to (130.center);
		\draw (131.center) to (124);
		\draw [in=90, out=-180, looseness=0.75] (135) to (123);
		\draw [in=-180, out=-90, looseness=0.75] (123) to (124);
		\draw [in=-75, out=180, looseness=0.50] (125) to (123);
		\draw [in=180, out=-45, looseness=0.50] (123) to (126);
		\draw [in=-15, out=-180, looseness=0.25] (127) to (123);
		\draw [in=180, out=15, looseness=0.25] (123) to (132);
		\draw [in=45, out=-180, looseness=0.50] (133) to (123);
		\draw [in=75, out=180, looseness=0.50] (134) to (123);
		\draw (166) to (167.center);
		\draw (168.center) to (165);
		\draw (164) to (169.center);
		\draw (170.center) to (163);
		\draw (158) to (159.center);
		\draw (160.center) to (157);
		\draw (156) to (161.center);
		\draw (162.center) to (155);
		\draw [in=90, out=-180, looseness=0.75] (166) to (154);
		\draw [in=-180, out=-90, looseness=0.75] (154) to (155);
		\draw [in=-75, out=180, looseness=0.50] (156) to (154);
		\draw [in=180, out=-45, looseness=0.50] (154) to (157);
		\draw [in=-15, out=-180, looseness=0.25] (158) to (154);
		\draw [in=180, out=15, looseness=0.25] (154) to (163);
		\draw [in=45, out=-180, looseness=0.50] (164) to (154);
		\draw [in=75, out=180, looseness=0.50] (165) to (154);
		\draw (80) to (166);
		\draw (79) to (165);
		\draw (78) to (164);
		\draw (77) to (163);
		\draw (72) to (158);
		\draw (71) to (157);
		\draw (70) to (156);
		\draw (69) to (155);
		\draw [style=geo] (189) to (187);
		\draw [style=geo] (187) to (188);
		\draw [style=geo] (188) to (190);
		\draw [style=geo] (190) to (189);
		\draw (191) to (187);
		\draw (191) to (183);
		\draw (189) to (191);
		\draw (191) to (185);
		\draw (191) to (186);
		\draw (190) to (191);
		\draw (191) to (184);
		\draw (191) to (188);
		\draw [in=75, out=-135] (185) to (194);
		\draw [in=-150, out=0] (194) to (186);
		\draw [in=-105, out=105, looseness=0.75] (194) to (183);
		\draw [in=-45, out=105, looseness=0.75] (193) to (188);
		\draw [in=345, out=150] (193) to (184);
		\draw [in=15, out=-135] (193) to (186);
		\draw [in=-165, out=-15] (190) to (193);
		\draw [in=120, out=-90] (192) to (187);
		\draw [in=-120, out=105, looseness=0.75] (183) to (192);
		\draw [in=0, out=150, looseness=0.75] (188) to (192);
		\draw [style=geo] (183) to (187);
		\draw [style=geo] (183) to (185);
		\draw [style=geo] (185) to (186);
		\draw [style=geo] (186) to (184);
		\draw [style=geo] (190) to (186);
		\draw [style=geo] (189) to (185);
		\draw [style=geo] (188) to (184);
		\draw [style=none] (195.center) to (192);
		\draw [style=none] (194) to (197.center);
		\draw [style=none] (193) to (196.center);
		\draw (198.center) to (187);
		\draw (199.center) to (183);
		\draw (200.center) to (189);
		\draw (201.center) to (185);
		\draw (202.center) to (186);
		\draw (203.center) to (190);
		\draw (204.center) to (184);
		\draw (205.center) to (188);
		\draw [style=geo] (184) to (183);
		\draw [in=-120, out=30] (194) to (184);
		\draw [in=120, out=-30] (192) to (184);
		\draw [style=geo] (208) to (209);
		\draw [style=geo] (206) to (208);
		\draw [style=geo] (209) to (207);
		\draw (208) to (210);
		\draw (210) to (207);
		\draw (209) to (210);
		\draw [style=geo] (213) to (214);
		\draw [style=geo] (212) to (211);
		\draw [style=geo] (211) to (213);
		\draw [style=geo] (214) to (212);
		\draw [style=geo] (206) to (211);
		\draw [style=geo] (213) to (208);
		\draw [style=geo] (209) to (214);
		\draw (207) to (215);
		\draw (215) to (214);
		\draw (209) to (215);
		\draw (215) to (212);
		\draw (206) to (216);
		\draw (216) to (212);
		\draw (207) to (216);
		\draw (216) to (211);
		\draw (206) to (217);
		\draw (217) to (214);
		\draw (208) to (217);
		\draw (217) to (212);
		\draw (209) to (217);
		\draw (217) to (211);
		\draw (218.center) to (208);
		\draw (219.center) to (209);
		\draw (223.center) to (214);
		\draw (220.center) to (207);
		\draw (221.center) to (206);
		\draw (222.center) to (213);
		\draw (225.center) to (211);
		\draw (224.center) to (212);
		\draw (226.center) to (210);
		\draw (216) to (229.center);
		\draw (228.center) to (227);
		\draw (217) to (213);
		\draw (217) to (207);
		\draw (210) to (206);
		\draw [style=geo] (207) to (206);
		\draw [style=geo] (207) to (212);
	\end{pgfonlayer}
\end{tikzpicture}

%% file: figures/ccz-nest.tikz
\begin{tikzpicture}
	\begin{pgfonlayer}{nodelayer}
		\node [style={gn_phase}] (0) at (2.25, 1) {$\frac{\pi}{4}$};
		\node [style={gn_phase}] (1) at (2.25, 0) {$\frac{\pi}{4}$};
		\node [style={gn_phase}] (2) at (2.25, -1) {$\frac{\pi}{4}$};
		\node [style=rn] (3) at (2.75, 0.5) {};
		\node [style=rn] (4) at (2.75, -0.5) {};
		\node [style=rn] (5) at (1.75, 0.5) {};
		\node [style=rn] (6) at (1.75, -0.5) {};
		\node [style=none] (7) at (1, 1) {};
		\node [style=none] (8) at (3.5, 1) {};
		\node [style=none] (9) at (1, 0) {};
		\node [style=none] (10) at (3.5, 0) {};
		\node [style=none] (11) at (1, -1) {};
		\node [style=none] (12) at (3.5, -1) {};
		\node [style={gn_phase}] (13) at (1.25, 0.5) {$\!\texttt{-}\frac{\pi}{4}\!$};
		\node [style={gn_phase}] (14) at (1.25, -0.5) {$\!\texttt{-}\frac{\pi}{4}\!$};
		\node [style={gn_phase}] (15) at (3.25, 0.5) {$\!\texttt{-}\frac{\pi}{4}\!$};
		\node [style={gn_phase}] (16) at (3.25, -0.5) {$\frac{\pi}{4}$};
		\node [style=none] (17) at (-2.5, 1) {};
		\node [style=none] (18) at (-1, 1) {};
		\node [style=none] (19) at (-2.5, 0) {};
		\node [style=none] (20) at (-1, 0) {};
		\node [style=none] (21) at (-2.5, -1) {};
		\node [style=none] (22) at (-1, -1) {};
		\node [style=none] (23) at (0, 0) {$=$};
		\node [style=gn] (24) at (-2, 1) {};
		\node [style=gn] (25) at (-2, 0) {};
		\node [style=gn] (26) at (-2, -1) {};
		\node [style=small hadamard] (27) at (-1.5, -0.5) {\small $\phantom{,}$};
	\end{pgfonlayer}
	\begin{pgfonlayer}{edgelayer}
		\draw (7.center) to (0);
		\draw (0) to (8.center);
		\draw (0) to (5);
		\draw (5) to (1);
		\draw (1) to (9.center);
		\draw (1) to (10.center);
		\draw (1) to (6);
		\draw (6) to (2);
		\draw (2) to (11.center);
		\draw (2) to (12.center);
		\draw (2) to (4);
		\draw (4) to (1);
		\draw (4) to (0);
		\draw (3) to (2);
		\draw (3) to (0);
		\draw (13) to (5);
		\draw (14) to (6);
		\draw (4) to (16);
		\draw (15) to (3);
		\draw (21.center) to (22.center);
		\draw (20.center) to (19.center);
		\draw (17.center) to (18.center);
		\draw (24) to (27);
		\draw (27) to (25);
		\draw (27) to (26);
	\end{pgfonlayer}
\end{tikzpicture}

%% file: msd.tex
\section{Proofs for Magic State Distillation and Injection}\label{sec:msdi}
\subsection{CCZ Magic State Distillation}\label{sec:cczmsd}
Here, we give a graphical portrayal of magic state distillation.
We have just derived that the logical CCZ gate in the $\llbracket 8,3,2 \rrbracket$ code is physically implemented by 8 $T$ and $T^\dagger$ gates:
\begin{equation}
    \tikzfigscale{1}{distillation-1}
\end{equation}
Taking the conjugate transpose of the above equation and inputting $\ket{+}^{\otimes 8}$, the ideal decoder acts on 8 $T$ and $T^\dagger$ magic states as:
\begin{equation}
    \tikzfigscale{1}{distillation-2}
\end{equation}
Substituting in the Z-X GPF of the encoder, the $\ket{+}$ states copy and fuse, leaving just the three-legged H-box.
\begin{equation}
    \tikzfigscale{0.8}{distillation-3}
\end{equation}
In other words, the decoding of 8 $T$ and $T^\dagger$ magic states results in a less noisy CCZ magic state.

Were it distilling the $\ket{0}$-controlled Z magic state from the inner CSS code of our qutrit $\llbracket 6,2,2 \rrbracket$ code, the proof steps are identical to the above.
Indeed, this proof applies for distilling magic states of any transversal diagonal gate in a CSS code whose X and Z logical operators are realized by physical X's and Z's respectively, i.e. whose encoder is a phase-free isometry in the ZX-calculus.

\tikzstyle{tikzfig}=[baseline=-0.25em, scale=0.5]
\subsection{CCZ Magic State Injection for Prime-Dimensional Qudits}
\begin{proposition}
    As is the case for qubits, for any odd prime qudit dimension, the Toffoli and CCZ gates are in the third level of the Clifford hierarchy.
\end{proposition}
\begin{proof}
  Noting the ZX diagram interpretations of the Z-basis states:
  \begin{equation}
    \tikzfigscale{1}{Kj-definition}
  \end{equation}
  Pauli Z directly commutes through the CCZ gate because both are diagonal in the Z basis.  We can commute Pauli X through the CCZ gate as follows:
  \begin{equation}
      \tikzfigscale{1}{ccz-x-push}
  \end{equation}
\end{proof}

As it is in the third level of the Clifford hierarchy, the Toffoli and CCZ gates admit magic state injection for any prime qudit dimension.
Moreover, both these gates are semi-Clifford~\cite{deSilvaN2021quditgatetp}:
\begin{definition}
    A gate $G$ is \emph{semi-Clifford} if it admits a decomposition $G = C_2 D C_1$, as a diagonal gate in some level of the Clifford hierarchy $D$, pre- and post-composed by Clifford gates $C_1$ and $C_2$ respectively.
\end{definition}
Therefore, a more efficient (halving the number of ancillae qudits) gate teleportation can be utilized~\cite{deSilvaN2021quditgatetp}:
\begin{proposition}
  For any prime qudit dimension, the CCZ gate is deterministically implementable by one round of the below magic state injection protocol, with one copy of the CCZ magic state.
\end{proposition}
\begin{proof}
  The multiqudit case of the teleportation circuit follows from parallelism of the singlequdit case.
  \begin{equation}
    \tikzfigscale{1}{Kj-teleport}
  \end{equation}
  By commuting the intended gate to teleport (the CCZ gate) through the three-qubit teleportation circuit, we get that the CCZ magic state can be teleported by that circuit, with Pauli X and Clifford CZ gate corrections controlled on the measurement outcomes.
  \begin{equation}
      \tikzfigscale{1}{ccz-teleport}
  \end{equation}
\end{proof}

\subsection{\texorpdfstring{$\ket{0}$-}{0-}Controlled Z Injection}

We start by deriving the ZX diagram for the $\ket{0}$-controlled Z gate, by conjugating the $\ket{0}$-controlled $X^\dagger$ from \cite[Equation 5]{Roy_2023}:
\begin{equation}
    \input{ZCZ/ZCZzh.tikz}
\end{equation}

Its magic state is obtained by inputting $\ket{++}$:
\begin{equation}
    \input{ZCZ/ZCZms.tikz}
\end{equation}

The standard injection circuit is below. When the measurement outcome is $\bra{00}$, i.e. both $a$ and $b$ are zero, then we have that this circuit implements the $\ket{0}$-controlled Z gate.
\begin{equation}
    \tikzfigscale{0.73}{ZCZ/ZCZuncorrected}
\end{equation}

To see if this can be deterministically injectable, we will need the following lemma:
\begin{lemma}
For $a,b \in \mathbb{Z}$,
\begin{equation}
    \input{ZCZ/ZCZlem.tikz}
\end{equation}
\end{lemma}

We can thereafter push both X corrections from measurement through the magic state, to obtain a diagonal Clifford correction.
\begin{equation}
    \input{ZCZ/ZCZlvl3.tikz}
\end{equation}
Pauli Z commutes with the $\ket{0}$-controlled gate because they're both gates diagonal in the Z-basis. Therefore, this proves that the $\ket{0}$-controlled gate is in the third level of the Clifford hierarchy, because pushing through any Pauli operator results in a correction that is Clifford, i.e. in the second level of the Clifford hierarchy.

We can then append to our earlier circuit the inverse of this error term, as a correction parametrized by the measurement outcome.
\begin{equation}\label{eq:ZCZcorrected}
    \tikzfigscale{0.95}{ZCZ/ZCZcorrected}
\end{equation}
As a result, irrespective of measurement outcome, we have deterministic implementation of the $\ket{0}$-controlled gate by magic state injection.

\subsection{Implementing the AND Gate via Magic State Distillation and Injection}\label{sec:andmsdi}
Up to now, we have been working with the simpler setting---injecting diagonal logical gates distilled in a CSS code whose logical X's and Z's are implemented by physical X's and Z's respectively.
Our goal is now to distill and inject the AND gate into any qutrit CSS code, so long as you have the capability to do fault-tolerant S gates (for instance, transversal in the code or via subsequent injection).
Essentially, we want the protocol to not require the code to have a transversal H gate, i.e. be self-dual.

In the two-qutrit unitary $U$ that fixes our logical operators such that the transversal operation is AND instead of $\ket{0}$-controlled Z, it is the CX gate that causes a problem with injecting the AND magic state directly such that it cannot be corrected as to be deterministic.
Instead, as done for $U^\dagger$ below, we push the CX out to be the outermost gate, in order to inject the remaining gates as part of the magic state. As we are injecting into a CSS code, the CX can be executed transversally conjugating the injection.
\begin{equation}
    \input{UwoCX.tikz}
\end{equation}

By decoding 3 T and 3 T$^\dagger$ states (which can be recursively distilled) in the code without this CX, we get a magic state for the AND gate up to this CX.
\begin{equation}
    \tikzfigscale{0.83}{ANDwoCXmsd}
\end{equation}

We can then perform deterministic injection, through proving this as a modification of our earlier protocol in Equation~\eqref{eq:ZCZcorrected}.
\begin{equation}
    \tikzfigscale{0.77}{ANDwoCXcorrected}
\end{equation}
The resulting correction is a Clifford unitary where all but one gate---the $Z(2b,2b)$ gate---are Pauli or CX gates, and hence transversally implementable on CSS codes~\cite{GottesmanD1997thesis}.
As a single-qutrit diagonal Clifford unitary, the $Z(2b,2b)$ gate can be implemented by the standard magic state injection protocol.

The remaining piece of the puzzle, is that the protocol measures two qutrits in two different bases, X and Z.
Presume we have access to the $\overline{\ket{0+}}$ logical state in the code as a resource.
If there are more than two logical qubits per code block, the resource state should be $\overline{\ket{0+...+}}$ or $\overline{\ket{0...0+}}$, depending on whether you want to measure one logical qubit in the Z or X basis respectively.

Say that we want to measure logical $\overline{\bra{a}\otimes I}$, where $I$ is the identity operation on all logical qutrits except the one we want to measure; without loss of generality, let us say we want to measure the first logical qutrit.
Observe that the $\overline{\ket{0+}} = E\ket{0+}$ state is representable as a phase-free ZX diagram, where the logical qutrit you want to measure is set to $\ket{0}$ and the rest to $\ket{+}$.
\begin{equation}
    \input{zeroplusmagicstate.tikz}
\end{equation}
The resulting ZX diagram is a bipartite graph whose connectivity is that of logical $Z_1$ and Z stabilizer generators of the code.
This correspondence between ZX diagrams of stabilizer codes, stabilizer tableaus, and Tanner graphs is detailed in~\cite{Yeh2026zxstabtab}.

Injecting this resource state yields the desired Z-basis measurement of the first logical qutrit.
\begin{equation}
    \tikzfigscale{0.78}{zeroplusinj}
\end{equation}
It additionally performs Steane-style syndrome measurement~\cite{Steane_1997} of the Z checks.

Exchanging the roles of Z and X, and $\ket{0}$ and $\ket{+}$, likewise enables measuring $\overline{I \otimes \bra{b}H}$ in addition to Steane-style syndrome measurement of the X checks.

%% file: figures/ZCZ/ZCZzh.tikz
\begin{tikzpicture}
	\begin{pgfonlayer}{nodelayer}
		\node [style=none] (34) at (10.25, 0.75) {};
		\node [style=none] (35) at (10.25, -0.75) {};
		\node [style=Z phase dot] (38) at (9.625, -0.75) {$\frac{2}{1}$};
		\node [style=hadamard] (42) at (8.625, 0) {};
		\node [style=box] (0) at (-9, -0.75) {};
		\node [style=box] (1) at (-9, -0.75) {$Z$};
		\node [style=vertex set] (2) at (-9, 0.75) {0};
		\node [style=none] (3) at (-10, 0.75) {};
		\node [style=none] (4) at (-8, 0.75) {};
		\node [style=none] (5) at (-8, -0.75) {};
		\node [style=none] (6) at (-10, -0.75) {};
		\node [style=none] (7) at (-7, 0) {$=$};
		\node [style=box] (8) at (-3.5, -0.75) {};
		\node [style=box] (9) at (-3.5, -0.75) {$\!X^\dagger\!$};
		\node [style=vertex set] (10) at (-3.5, 0.75) {0};
		\node [style=none] (11) at (-6, 0.75) {};
		\node [style=none] (12) at (-1, 0.75) {};
		\node [style=none] (13) at (-1, -0.75) {};
		\node [style=none] (14) at (-6, -0.75) {};
		\node [style=box] (15) at (-5, -0.75) {$\!H\!$};
		\node [style=box] (16) at (-2, -0.75) {$\!H^\dagger\!$};
		\node [style=none] (17) at (0, 0) {$=$};
		\node [style=none] (21) at (1, 0.75) {};
		\node [style=none] (22) at (5.25, 0.75) {};
		\node [style=none] (23) at (5.25, -0.75) {};
		\node [style=none] (24) at (1, -0.75) {};
		\node [style=Z dot] (25) at (1.5, 0.75) {};
		\node [style=X phase dot] (27) at (4, -0.75) {$\frac{1}{2}$};
		\node [style=hadamard] (28) at (1.5, -0.75) {\tiny 2};
		\node [style=hadamard] (29) at (4.75, -0.75) {\tiny 1};
		\node [style=hadamard] (30) at (2.375, 0) {\tiny 1};
		\node [style=hadamard] (31) at (3, 0) {\tiny 2};
		\node [style=none] (32) at (6.25, 0) {$=$};
		\node [style=none] (33) at (7.25, 0.75) {};
		\node [style=none] (36) at (7.25, -0.75) {};
		\node [style=Z dot] (37) at (7.75, 0.75) {};
		\node [style=hadamard] (41) at (8.625, 0) {\tiny 1};
	\end{pgfonlayer}
	\begin{pgfonlayer}{edgelayer}
		\draw [style=simple] (2) to (1);
		\draw [style=simple] (3.center) to (2);
		\draw [style=simple] (2) to (4.center);
		\draw [style=simple] (6.center) to (1);
		\draw [style=simple] (1) to (5.center);
		\draw [style=simple] (10) to (9);
		\draw [style=simple] (11.center) to (10);
		\draw [style=simple] (10) to (12.center);
		\draw [style=simple] (14.center) to (9);
		\draw [style=simple] (9) to (13.center);
		\draw [in=150, out=-30] (25) to (30);
		\draw [in=-60, out=-150] (30) to (25);
		\draw (21.center) to (22.center);
		\draw (23.center) to (24.center);
		\draw [in=-180, out=-15, looseness=1.25] (31) to (27);
		\draw (30) to (31);
		\draw [in=150, out=-30] (37) to (41);
		\draw [in=-60, out=-150] (41) to (37);
		\draw (33.center) to (34.center);
		\draw (35.center) to (36.center);
		\draw [in=-180, out=-15, looseness=1.25] (42) to (38);
	\end{pgfonlayer}
\end{tikzpicture}

%% file: figures/ZCZ/ZCZms.tikz
\begin{tikzpicture}
	\begin{pgfonlayer}{nodelayer}
		\node [style=none] (34) at (1.625, 0.75) {};
		\node [style=none] (35) at (1.625, -0.75) {};
		\node [style=Z phase dot] (38) at (1, -0.75) {$\frac{2}{1}$};
		\node [style=none] (4) at (-4.375, 0) {$|ZCZ\rangle$};
		\node [style=none] (7) at (-2.625, 0) {$=$};
		\node [style=Z dot] (33) at (-1.625, 0.75) {};
		\node [style=Z dot] (36) at (-1.625, -0.75) {};
		\node [style=Z dot] (37) at (-0.875, 0.75) {};
		\node [style=hadamard] (41) at (0, 0) {\tiny 1};
		\node [style=none] (43) at (4.875, 0.75) {};
		\node [style=none] (44) at (4.875, -0.75) {};
		\node [style=Z phase dot] (45) at (4.25, -0.75) {$\frac{2}{1}$};
		\node [style=none] (47) at (2.625, 0) {$=$};
		\node [style=Z dot] (52) at (4.25, 0.75) {};
		\node [style=hadamard] (53) at (3.75, 0) {\tiny 1};
	\end{pgfonlayer}
	\begin{pgfonlayer}{edgelayer}
		\draw [in=150, out=-30] (37) to (41);
		\draw [in=-60, out=-150] (41) to (37);
		\draw (33) to (34.center);
		\draw (35.center) to (36);
		\draw [bend left=75] (52) to (53);
		\draw [bend left=60] (53) to (52);
		\draw (43.center) to (52);
		\draw [bend right=45] (53) to (45);
		\draw [in=-180, out=-15, looseness=1.25] (41) to (38);
		\draw (44.center) to (45);
	\end{pgfonlayer}
\end{tikzpicture}

%% file: figures/ZCZ/ZCZlem.tikz
\begin{tikzpicture}
	\begin{pgfonlayer}{nodelayer}
		\node [style=X phase dot] (0) at (-6.75, 0) {$\frac{a}{2a}$};
		\node [style=hadamard] (1) at (-5.75, 0) {\tiny $b$};
		\node [style=none] (2) at (-7.5, 0) {};
		\node [style=none] (3) at (-5, 0) {};
		\node [style=none] (4) at (-4, 0) {$=$};
		\node [style=Z phase dot] (5) at (-0.375, 0) {$\frac{2a}{a}$};
		\node [style=hadamard] (6) at (1.25, 0) {\tiny $1$};
		\node [style=none] (7) at (-2, 0) {};
		\node [style=none] (8) at (2, 0) {};
		\node [style=none] (9) at (3, 0) {$=$};
		\node [style=hadamard] (10) at (-1.25, 0) {\tiny $1$};
		\node [style=none] (13) at (-2.625, 1.875) {};
		\node [style=none] (14) at (-2.625, -1.875) {};
		\node [style=none] (15) at (-3, 0) {};
		\node [style=X phase dot] (16) at (-0.375, 1.75) {$\frac{a}{2a}$};
		\node [style=none] (18) at (-2, 1.75) {};
		\node [style=none] (19) at (2, 1.75) {};
		\node [style=none] (20) at (3, 1.75) {$=$};
		\node [style=Z dot] (21) at (0.5, 1.75) {};
		\node [style=Z dot] (22) at (1.25, 1.75) {};
		\node [style=hadamard] (23) at (0.5, 0) {\tiny $2$};
		\node [style=Z phase dot] (24) at (-0.375, -1.75) {$\frac{a}{2a}$};
		\node [style=hadamard] (25) at (1.25, -1.75) {\tiny $2$};
		\node [style=none] (26) at (-2, -1.75) {};
		\node [style=none] (27) at (2, -1.75) {};
		\node [style=none] (28) at (3, -1.75) {$=$};
		\node [style=hadamard] (29) at (-1.25, -1.75) {\tiny $2$};
		\node [style=hadamard] (30) at (0.5, -1.75) {\tiny $1$};
		\node [style=none] (32) at (4, 1.75) {};
		\node [style=none] (33) at (6.375, 1.75) {};
		\node [style=hadamard] (36) at (4.75, 1.75) {\tiny $0$};
		\node [style=Z phase dot] (37) at (5.625, 0) {$\frac{2a}{a}$};
		\node [style=none] (39) at (4, 0) {};
		\node [style=none] (40) at (6.375, 0) {};
		\node [style=hadamard] (41) at (4.75, 0) {\tiny $1$};
		\node [style=Z phase dot] (42) at (5.625, -1.75) {$\frac{a}{2a}$};
		\node [style=none] (44) at (4, -1.75) {};
		\node [style=none] (45) at (6.375, -1.75) {};
		\node [style=hadamard] (46) at (4.75, -1.75) {\tiny $2$};
		\node [style=Z phase dot] (47) at (5.625, 1.75) {$\frac{0}{0}$};
		\node [style=none] (48) at (7.25, 1.75) {\scriptsize ,  $b=0$};
		\node [style=none] (49) at (7.25, 0) {\scriptsize ,  $b=1$};
		\node [style=none] (50) at (7.25, -1.75) {\scriptsize ,  $b=2$};
		\node [style=none] (51) at (9.625, 0) {$=$};
		\node [style=none] (52) at (8.25, 1.875) {};
		\node [style=none] (53) at (8.25, -1.875) {};
		\node [style=none] (54) at (8.625, 0) {};
		\node [style=none] (55) at (10.625, 0) {};
		\node [style=none] (56) at (13, 0) {};
		\node [style=hadamard] (57) at (11.375, 0) {\tiny $b$};
		\node [style=Z phase dot] (58) at (12.25, 0) {$\frac{2ab}{ab}$};
	\end{pgfonlayer}
	\begin{pgfonlayer}{edgelayer}
		\draw (3.center) to (2.center);
		\draw (8.center) to (7.center);
		\draw [style=less thin] (13.center)
			 to [in=15, out=180, looseness=0.75] (15.center)
			 to [in=180, out=-15, looseness=0.75] (14.center);
		\draw (18.center) to (21);
		\draw (22) to (19.center);
		\draw (27.center) to (26.center);
		\draw (32.center) to (33.center);
		\draw (40.center) to (39.center);
		\draw (45.center) to (44.center);
		\draw [style=less thin, in=165, out=0, looseness=0.75] (52.center) to (54.center);
		\draw [style=less thin, in=0, out=-165, looseness=0.75] (54.center) to (53.center);
		\draw (55.center) to (56.center);
	\end{pgfonlayer}
\end{tikzpicture}

%% file: figures/ZCZ/ZCZlvl3.tikz
\begin{tikzpicture}
	\begin{pgfonlayer}{nodelayer}
		\node [style=Z phase dot] (0) at (-20, -1.625) {$\frac{2}{1}$};
		\node [style=Z dot] (1) at (-20, 1.375) {};
		\node [style=hadamard] (2) at (-22, 0) {\tiny 1};
		\node [style=X phase dot] (3) at (-19.125, 1.375) {$\frac{2a}{a}$};
		\node [style=X phase dot] (4) at (-19.125, -1.625) {$\frac{2b}{b}$};
		\node [style=none] (5) at (-18.5, 1.375) {};
		\node [style=none] (6) at (-18.5, -1.625) {};
		\node [style=none] (7) at (-17.5, 0) {$=$};
		\node [style=Z phase dot] (8) at (-13.75, -1.625) {$\frac{2}{1}$};
		\node [style=Z dot] (9) at (-13.75, 1.375) {};
		\node [style=hadamard] (10) at (-16.25, 0) {\tiny 1};
		\node [style=X phase dot] (11) at (-15.45, 1.375) {$\frac{2a}{a}$};
		\node [style=X phase dot] (12) at (-12.875, -1.625) {$\frac{2b}{b}$};
		\node [style=none] (13) at (-12.25, 1.375) {};
		\node [style=none] (14) at (-12.25, -1.625) {};
		\node [style=none] (15) at (-11.25, 0) {$=$};
		\node [style=X phase dot] (16) at (-14.5, 0.25) {$\frac{2a}{a}$};
		\node [style=Z phase dot] (40) at (-6.5, -1.625) {$\frac{2}{1}$};
		\node [style=Z dot] (41) at (-6.5, 1.375) {};
		\node [style=hadamard] (42) at (-9.5, 0) {\tiny 1};
		\node [style=X phase dot] (44) at (-5.625, -1.625) {$\frac{2b}{b}$};
		\node [style=none] (45) at (-5, 1.375) {};
		\node [style=none] (46) at (-5, -1.625) {};
		\node [style=none] (47) at (-4, 0) {$=$};
		\node [style=X phase dot] (48) at (-7.25, 0.25) {$\frac{2a}{a}$};
		\node [style=Z dot] (50) at (-8, -1.6) {};
		\node [style=hadamard] (51) at (-8, -0.825) {\scriptsize $a$};
		\node [style=Z dot] (52) at (-8, -0.125) {};
		\node [style=Z phase dot] (53) at (0.5, -1.625) {$\frac{2}{1}$};
		\node [style=Z dot] (54) at (0.5, 1.375) {};
		\node [style=hadamard] (55) at (-2.5, 0) {\tiny 1};
		\node [style=X phase dot] (56) at (1.375, -1.625) {$\frac{2b}{b}$};
		\node [style=none] (57) at (2, 1.375) {};
		\node [style=none] (58) at (2, -1.625) {};
		\node [style=none] (59) at (-17.5, -5) {$=$};
		\node [style=X phase dot] (60) at (-1.75, -0.125) {$\frac{2a}{a}$};
		\node [style=Z phase dot] (61) at (-1, -1.6) {$\frac{2a\!^2}{a\!^2}$};
		\node [style=hadamard] (62) at (-1, -0.825) {\scriptsize $a$};
		\node [style=Z dot] (63) at (-1, -0.125) {};
		\node [style=none] (64) at (-12.875, -6.625) {};
		\node [style=Z dot] (65) at (-13, -3.625) {};
		\node [style=hadamard] (66) at (-16, -5) {\tiny 1};
		\node [style=X phase dot] (67) at (-13.75, -6.625) {$\frac{2b}{b}$};
		\node [style=none] (68) at (-11.375, -3.625) {};
		\node [style=none] (69) at (-11.375, -6.625) {};
		\node [style=none] (70) at (-10.375, -5) {$=$};
		\node [style=Z dot] (72) at (-14.5, -6.6) {};
		\node [style=hadamard] (73) at (-14.5, -5.825) {\scriptsize $a$};
		\node [style=Z dot] (74) at (-14.5, -5.125) {};
		\node [style=Z dot] (75) at (-15.25, -6.35) {};
		\node [style=hadamard] (76) at (-15.25, -5.575) {\scriptsize $a$};
		\node [style=Z dot] (77) at (-15.25, -3.625) {};
		\node [style=Z phase dot] (78) at (-12, -6.6) {$\frac{2a\!^2}{a\!^2}$};
		\node [style=Z phase dot] (79) at (-12.875, -6.625) {$\frac{2}{1}$};
		\node [style=none] (80) at (-5.75, -6.625) {};
		\node [style=hadamard] (82) at (-8.875, -5) {\tiny 1};
		\node [style=X phase dot] (83) at (-6.625, -6.625) {$\frac{2b}{b}$};
		\node [style=none] (84) at (-4.25, -3.625) {};
		\node [style=none] (85) at (-4.25, -6.625) {};
		\node [style=Z dot] (87) at (-7.375, -6.6) {};
		\node [style=hadamard] (88) at (-7.375, -5.575) {\scriptsize $2a$};
		\node [style=Z dot] (89) at (-7.375, -3.625) {};
		\node [style=Z phase dot] (93) at (-4.875, -6.6) {$\frac{2a\!^2}{a\!^2}$};
		\node [style=Z phase dot] (94) at (-5.75, -6.625) {$\frac{2}{1}$};
		\node [style=none] (95) at (-3.25, -5) {$=$};
		\node [style=none] (96) at (0.625, -6.625) {};
		\node [style=hadamard] (98) at (-1.75, -5) {\tiny 1};
		\node [style=X phase dot] (99) at (-1, -6.375) {$\frac{2b}{b}$};
		\node [style=none] (100) at (2.125, -3.625) {};
		\node [style=none] (101) at (2.125, -6.625) {};
		\node [style=Z dot] (102) at (-0.25, -6.6) {};
		\node [style=hadamard] (103) at (-0.25, -5.575) {\scriptsize $2a$};
		\node [style=Z phase dot] (104) at (-0.25, -3.625) {$\frac{ab}{2ab}$};
		\node [style=Z phase dot] (105) at (1.5, -6.6) {$\frac{2a\!^2}{a\!^2}$};
		\node [style=Z phase dot] (106) at (0.625, -6.625) {$\frac{2}{1}$};
		\node [style=none] (107) at (-10.375, -10.5) {$=$};
		\node [style=none] (108) at (-6.5, -12.125) {};
		\node [style=hadamard] (109) at (-8.875, -10.5) {\tiny 1};
		\node [style=none] (111) at (-5, -9.125) {};
		\node [style=none] (112) at (-5, -12.125) {};
		\node [style=Z dot] (113) at (-7.375, -12.1) {};
		\node [style=hadamard] (114) at (-7.375, -11.075) {\scriptsize $2a$};
		\node [style=Z dot] (115) at (-7.375, -9.125) {};
		\node [style=Z phase dot] (116) at (-5.625, -12.1) {$\frac{2a\!^2}{a\!^2}$};
		\node [style=Z phase dot] (117) at (-6.5, -12.125) {$\frac{2}{1}$};
		\node [style=Z phase dot] (118) at (-5.625, -9.125) {$\frac{ab}{2ab}$};
		\node [style=hadamard] (119) at (-7.375, -8.325) {\scriptsize $b$};
		\node [style=none] (120) at (-4, -10.5) {$=$};
		\node [style=none] (121) at (-1, -12.125) {};
		\node [style=hadamard] (122) at (-2.5, -10.5) {\tiny 1};
		\node [style=none] (123) at (1.375, -9.125) {};
		\node [style=none] (124) at (1.375, -12.125) {};
		\node [style=Z dot] (127) at (-1, -9.125) {};
		\node [style=Z phase dot] (128) at (0.75, -12.1) {$\frac{2a\!^2}{a\!^2}$};
		\node [style=Z phase dot] (129) at (-1, -12.125) {$\frac{2}{1}$};
		\node [style=Z phase dot] (130) at (0.75, -9.125) {$\frac{ab}{2ab}$};
		\node [style=Z dot] (131) at (-0.125, -12.1) {};
		\node [style=hadamard] (132) at (-0.125, -10.5) {\scriptsize $2a$};
		\node [style=Z phase dot] (133) at (-0.125, -9.125) {$\frac{b}{b}$};
		\node [style=none] (134) at (-17.5, -10.5) {$=$};
		\node [style=none] (135) at (-13.625, -12.125) {};
		\node [style=hadamard] (136) at (-16, -10.5) {\tiny 1};
		\node [style=none] (138) at (-12.125, -9.125) {};
		\node [style=none] (139) at (-12.125, -12.125) {};
		\node [style=Z dot] (140) at (-14.5, -12.1) {};
		\node [style=hadamard] (141) at (-14.5, -11.075) {\scriptsize $2a$};
		\node [style=Z phase dot] (142) at (-14.5, -9.125) {$\frac{ab}{2ab}$};
		\node [style=Z phase dot] (143) at (-12.75, -12.1) {$\frac{2a\!^2}{a\!^2}$};
		\node [style=Z phase dot] (144) at (-13.625, -12.125) {$\frac{2}{1}$};
		\node [style=Z dot] (145) at (-15.25, -10.35) {};
		\node [style=hadamard] (146) at (-15.25, -9.725) {\scriptsize $b$};
		\node [style=Z dot] (147) at (-15.25, -9.125) {};
	\end{pgfonlayer}
	\begin{pgfonlayer}{edgelayer}
		\draw [bend left=45] (1) to (2);
		\draw [bend left=60] (2) to (1);
		\draw [in=180, out=-90] (2) to (0);
		\draw (1) to (3);
		\draw (4) to (0);
		\draw (5.center) to (3);
		\draw (6.center) to (4);
		\draw [bend left=45] (9) to (10);
		\draw [bend left=60] (10) to (9);
		\draw [in=-180, out=-75] (10) to (8);
		\draw (12) to (8);
		\draw (14.center) to (12);
		\draw (13.center) to (9);
		\draw [bend left=45] (41) to (42);
		\draw [bend left=60] (42) to (41);
		\draw [in=-180, out=-75, looseness=1.25] (42) to (40);
		\draw (44) to (40);
		\draw (46.center) to (44);
		\draw (45.center) to (41);
		\draw (52) to (50);
		\draw [bend left=45] (54) to (55);
		\draw [bend left=60] (55) to (54);
		\draw [in=-180, out=-75, looseness=1.25] (55) to (53);
		\draw (56) to (53);
		\draw (58.center) to (56);
		\draw (57.center) to (54);
		\draw (63) to (61);
		\draw [bend left=45] (65) to (66);
		\draw [bend left=60] (66) to (65);
		\draw (69.center)
			 to (64.center)
			 to [in=-75, out=-180, looseness=1.25] (66);
		\draw (68.center) to (65);
		\draw (74) to (72);
		\draw (77) to (75);
		\draw (85.center)
			 to (80.center)
			 to [in=-75, out=-180, looseness=1.25] (82);
		\draw (89) to (87);
		\draw (101.center) to (96.center);
		\draw [in=-75, out=-180, looseness=1.25] (96.center) to (98);
		\draw (104) to (102);
		\draw [bend left=45] (89) to (82);
		\draw [bend left=60] (82) to (89);
		\draw (84.center) to (89);
		\draw [bend left=45] (104) to (98);
		\draw [bend left=60] (98) to (104);
		\draw (100.center) to (104);
		\draw (112.center) to (108.center);
		\draw [in=-75, out=-180, looseness=1.25] (108.center) to (109);
		\draw (115) to (113);
		\draw [bend left=45] (115) to (109);
		\draw [bend left=60] (109) to (115);
		\draw (111.center) to (115);
		\draw [in=0, out=30, looseness=1.50] (115) to (119);
		\draw [in=150, out=180, looseness=1.50] (119) to (115);
		\draw (124.center) to (121.center);
		\draw [in=-90, out=-180, looseness=0.75] (121.center) to (122);
		\draw [bend left=45] (127) to (122);
		\draw [bend left=60] (122) to (127);
		\draw (123.center) to (127);
		\draw (133) to (131);
		\draw (139.center) to (135.center);
		\draw [in=-75, out=-180, looseness=1.25] (135.center) to (136);
		\draw (142) to (140);
		\draw [bend left=45] (142) to (136);
		\draw [bend left=60] (136) to (142);
		\draw (138.center) to (142);
		\draw (147) to (145);
	\end{pgfonlayer}
\end{tikzpicture}

%% file: figures/UwoCX.tikz
\begin{tikzpicture}
	\begin{pgfonlayer}{nodelayer}
		\node [style=none] (126) at (-8.5, 1) {};
		\node [style=none] (127) at (-8.5, -1) {};
		\node [style=X phase dot] (129) at (-7.25, 1) {$\frac{2}{1}$};
		\node [style=hadamard] (130) at (-7.25, -1) {\scriptsize 2};
		\node [style=X dot] (131) at (-6.25, 1) {};
		\node [style=Z dot] (132) at (-5.5, -1) {};
		\node [style=dualizer] (133) at (-4.75, -1) {\scriptsize D};
		\node [style=none] (134) at (-4, 1) {};
		\node [style=none] (135) at (-4, -1) {};
		\node [style=none] (136) at (-3, 0) {$=$};
		\node [style=none] (149) at (-2, 1) {};
		\node [style=none] (150) at (-2, -1) {};
		\node [style=X phase dot] (152) at (-0.75, 1) {$\frac{2}{1}$};
		\node [style=hadamard] (153) at (-0.75, -1) {\scriptsize 2};
		\node [style=X dot] (154) at (0.25, 1) {};
		\node [style=Z dot] (155) at (1, -1) {};
		\node [style=dualizer] (156) at (0, -1) {\scriptsize D};
		\node [style=none] (157) at (1.75, 1) {};
		\node [style=none] (158) at (1.75, -1) {};
		\node [style=dualizer] (337) at (0.625, 0) {\scriptsize D};
		\node [style=none] (338) at (2.75, 0) {$=$};
		\node [style=none] (339) at (3.75, 1) {};
		\node [style=none] (340) at (3.75, -1) {};
		\node [style=X phase dot] (341) at (5, 1) {$\frac{2}{1}$};
		\node [style=hadamard] (342) at (5, -1) {\scriptsize 1};
		\node [style=X dot] (343) at (6.75, 1) {};
		\node [style=Z dot] (344) at (6, -1) {};
		\node [style=none] (346) at (7.5, 1) {};
		\node [style=none] (347) at (7.5, -1) {};
	\end{pgfonlayer}
	\begin{pgfonlayer}{edgelayer}
		\draw (126.center) to (134.center);
		\draw (127.center) to (135.center);
		\draw (131) to (132);
		\draw (149.center) to (157.center);
		\draw (150.center) to (158.center);
		\draw (154) to (155);
		\draw (339.center) to (346.center);
		\draw (340.center) to (347.center);
		\draw (343) to (344);
	\end{pgfonlayer}
\end{tikzpicture}

%% file: figures/zeroplusmagicstate.tikz
\begin{tikzpicture}
	\begin{pgfonlayer}{nodelayer}
		\node [style=X dot] (1) at (-2.75, 2) {};
		\node [style=X dot] (2) at (-1.75, 2) {};
		\node [style=X dot] (4) at (-1.75, 0) {};
		\node [style=X dot] (5) at (-1.75, -1.25) {};
		\node [style=Z dot] (6) at (0, 2) {};
		\node [style=Z dot] (7) at (0, -1.5) {};
		\node [style=none] (8) at (0.75, 2) {};
		\node [style=none] (9) at (0.75, -1.5) {};
		\node [style=none] (10) at (-0.625, 0.5) {\rotatebox{90}{$...$}};
		\node [style=none] (11) at (-0.75, 1) {};
		\node [style=none] (12) at (-0.75, 0) {};
		\node [style=none] (13) at (-1, 0.5) {};
		\node [style=none] (14) at (-1, 0) {};
		\node [style=none] (15) at (-0.5, 1) {};
		\node [style=none] (16) at (-0.5, 0) {};
		\node [style=none] (17) at (0.5, 0.5) {\rotatebox{90}{$...$}};
		\node [style=none] (18) at (-1.75, -0.625) {\rotatebox{90}{$...$}};
		\node [style=none] (19) at (1.75, 0) {$=$};
		\node [style=X dot] (22) at (3, 2) {};
		\node [style=X dot] (24) at (3, 0.5) {};
		\node [style=X dot] (25) at (3, -0.75) {};
		\node [style=Z dot] (26) at (4.75, 2) {};
		\node [style=Z dot] (27) at (4.75, -1.5) {};
		\node [style=none] (28) at (5.5, 2) {};
		\node [style=none] (29) at (5.5, -1.5) {};
		\node [style=none] (30) at (4.125, 0.5) {\rotatebox{90}{$...$}};
		\node [style=none] (31) at (4, 1) {};
		\node [style=none] (32) at (4, 0) {};
		\node [style=none] (34) at (3.75, 0.25) {};
		\node [style=none] (35) at (4.25, 1) {};
		\node [style=none] (36) at (4.25, 0) {};
		\node [style=none] (37) at (5.25, 0.5) {\rotatebox{90}{$...$}};
		\node [style=none] (38) at (3, -0.125) {\rotatebox{90}{$...$}};
		\node [style=none] (39) at (5.5, 2.75) {\scriptsize $\overline{Z_1}$ connectivity};
		\node [style=none] (40) at (3.5, 2.5) {};
		\node [style=none] (41) at (3.25, 2.25) {};
		\node [style=none] (42) at (3.25, 2.5) {};
		\node [style=none] (43) at (3.5, 2.25) {};
		\node [style=none] (44) at (5.5, -0.375) {\scriptsize Z stabilizers'};
		\node [style=none] (45) at (3.875, -0.25) {};
		\node [style=none] (46) at (3.125, -0.125) {};
		\node [style=none] (47) at (3.25, -0.375) {};
		\node [style=none] (48) at (3.375, 0) {};
		\node [style=none] (49) at (5.5, -0.825) {\scriptsize connectivity};
		\node [style=Z dot] (50) at (-7.75, 1) {};
		\node [style=X dot] (51) at (-7.75, 2) {};
		\node [style=X dot] (52) at (-6.75, 2) {};
		\node [style=X dot] (53) at (-6.75, 1) {};
		\node [style=X dot] (54) at (-6.75, 0) {};
		\node [style=X dot] (55) at (-6.75, -1.25) {};
		\node [style=Z dot] (56) at (-5, 2) {};
		\node [style=Z dot] (57) at (-5, -1.5) {};
		\node [style=none] (58) at (-4.25, 2) {};
		\node [style=none] (59) at (-4.25, -1.5) {};
		\node [style=none] (60) at (-5.625, 0.5) {\rotatebox{90}{$...$}};
		\node [style=none] (61) at (-5.75, 1) {};
		\node [style=none] (62) at (-5.75, 0) {};
		\node [style=none] (63) at (-6.125, 1) {};
		\node [style=none] (64) at (-6, 0) {};
		\node [style=none] (65) at (-5.5, 1) {};
		\node [style=none] (66) at (-5.5, 0) {};
		\node [style=none] (67) at (-4.5, 0.5) {\rotatebox{90}{$...$}};
		\node [style=none] (68) at (-6.75, -0.625) {\rotatebox{90}{$...$}};
		\node [style=none] (69) at (-3.5, 0) {$=$};
		\node [style=none] (70) at (-6.5, 0.5) {};
		\node [style=none] (71) at (-6.3, 0.75) {\rotatebox{45}{$...$}};
		\node [style=none] (74) at (-1.125, 1) {};
		\node [style=none] (75) at (-1.5, 0.5) {};
		\node [style=none] (76) at (-1.3, 0.75) {\rotatebox{45}{$...$}};
		\node [style=Z dot] (77) at (-1.5, 1.25) {};
		\node [style=Z dot] (78) at (-2, 0.75) {};
	\end{pgfonlayer}
	\begin{pgfonlayer}{edgelayer}
		\draw (1) to (2);
		\draw (2) to (11.center);
		\draw (4) to (14.center);
		\draw (12.center) to (5);
		\draw (16.center) to (7);
		\draw (7) to (9.center);
		\draw (6) to (15.center);
		\draw (6) to (8.center);
		\draw (22) to (31.center);
		\draw (24) to (34.center);
		\draw (32.center) to (25);
		\draw (36.center) to (27);
		\draw (27) to (29.center);
		\draw (26) to (35.center);
		\draw (26) to (28.center);
		\draw (40.center) to (41.center);
		\draw (42.center)
			 to (41.center)
			 to (43.center);
		\draw (45.center) to (46.center);
		\draw (47.center)
			 to (46.center)
			 to (48.center);
		\draw (51) to (52);
		\draw (50) to (53);
		\draw (52) to (61.center);
		\draw (63.center) to (53);
		\draw (54) to (64.center);
		\draw (62.center) to (55);
		\draw (66.center) to (57);
		\draw (57) to (59.center);
		\draw (56) to (65.center);
		\draw (56) to (58.center);
		\draw (53) to (70.center);
		\draw (77) to (74.center);
		\draw (78) to (75.center);
	\end{pgfonlayer}
\end{tikzpicture}

%% file: qubitemulation.tex
\section{Gate Counts for Qutrit Emulation of Qubit and Binary Operations}\label{sec:emuappendix}

\subsection{Emulating the qubit X, Z, S, CX and CZ gates}
\label{sec:XZSCXCZ}
The qubit Hadamard and CNOT operations cannot be emulated utilizing qutrit Clifford gates~\cite{bocharov_factoring_2017}. 
Moreover, we can show that the qubit Hadamard gate cannot be exactly and deterministically emulated using qutrit Clifford+T, despite this being an approximately universal gate set.
This follows from normal forms for single-qutrit Clifford+T operators, which impose the necessary but not sufficient condition that the matrix entries belong to the ring $\mathbb{Z}\left[\frac{1}{1 - \zeta}\right]$, as independently determined by Prakash, Jain, Kapur, and Seth~\cite{PrakashS2018normalform}; and Glaudell, Ross, and Taylor~\cite{GlaudellA2019canonical}:
\begin{definition}[\texorpdfstring{$\mathbb{Z}\left[\frac{1}{1 - \zeta}\right]$}{Single-qutrit Clifford + T normal form ring}]
    \label{def:qutritnormalformring}
    Let $\zeta=e^{i\frac{2\pi}{9}}$. The matrix elements of all single-qutrit Clifford+T unitaries must be elements of the ring
    \begin{equation}
        \mathbb{Z}\left[\frac{1}{1 - \zeta}\right] = \left\{\frac{A}{(1 - \zeta)^k} \middle| A \in \mathbb{Z}\left[\zeta\right], k \in \mathbb{N}\right\}
    \end{equation}
    where
    \begin{equation}
        \mathbb{Z}\left[\zeta\right] = \left\{\sum_{i = 0}^5 a_i (\zeta)^i \middle| \forall i, a_i \in \mathbb{Z}\right\}
    \end{equation}  
\end{definition}
Because $\frac{1}{\sqrt{2}}$ is not in the ring in Definition~\ref{def:qutritnormalformring}, it is impossible to exactly and deterministically emulate the qubit Hadamard gate utilizing single-qutrit Clifford+T gates.  In general, the qubit Clifford operations do not neatly embed in the qutrit Clifford fragment.

In spite of this, we can still emulate quite a few qubit gates using qutrits.
The qubit X gate can be emulated by the qutrit Clifford $X_{01}$ gate.  The qubit Z gate can be emulated by Clifford equivalence to the qutrit R gate:
\begin{equation}
    \text{qubit } Z \emulates X_{12} \text{ R } X_{12}
\end{equation}

Likewise, the qubit S gate can be emulated by Clifford equivalence to the qutrit $\sqrt{R}$ gate:
\begin{equation}
    \text{qubit } S \emulates X_{12} \sqrt{R} X_{12}
\end{equation}
Extending magic state injection of the $R$ gate~\cite{anwar_qutrit_2012}, despite being nowhere in the Clifford hierarchy, it can be shown that the $\sqrt{R} = \text{diag}(1,1,i)$ gate admits implementation by a repeat-until-success injection protocol of the $(1,1,i)^\top$ state, where the probability of success rapidly converges to 1 with the number of rounds.

\begin{lemma}[\cite{bocharov_factoring_2017}]
    \label{lemma:qubitcx}
    The qubit CX gate can be emulated by qutrit Clifford+T gates unitarily without ancillae, with T-count 6 and Clifford CX-count 8.
\end{lemma}
\begin{proof}
    For reference, the decomposition by Bocharov, Roetteler, and Svor~\cite[Proposition 4]{bocharov_factoring_2017} is given as a circuit diagram below:
    \begin{equation}
        \input{qutrit/qubitcqubitx.tikz}
    \end{equation}
    Like the qubit SWAP gate, the qutrit SWAP gate can be decomposed as three of the same two-qutrit gate applied successively in alternating directions.  By optimizing thereafter, the Clifford CX-count can be further reduced to 8.
\end{proof}


\begin{lemma}
    \label{lem:qubitcz}
    Utilizing one R gate, we can exactly emulate the qubit CZ gate.
\end{lemma}
\begin{proof} 
    \begin{equation}
        \label{eq:qubitcz}
        \input{qutrit/qubitcz.tikz}
    \end{equation}
    On qubit input states $\ket{00}$, $\ket{01}$, and $\ket{10}$, the states' sum modulo 3 is 0 or 1, and so no phase factor is applied.
    When the input state is $\ket{11}$, the states' sum modulo 3 is 2, and so a phase factor of $-1$ is applied.
\end{proof}
This improves on the three R-count decomposition achieved by conjugating the control of the $\ket{2}$-controlled qubit Z gate from~\cite{wetering_building_2023} by $X_{12}$.


\subsection{The qubit \texorpdfstring{$\ket{+}$}{\|+\>} and \texorpdfstring{$\ket{-}$}{\|-\>} states}
The qubit $\ket{+}$ state is qutrit Clifford equivalent to the Norell state, for which two magic state distillation protocols have been developed by~\cite{DawkinsH2015qutritmsd,PrakashS2020qutritmsd}.  Below, we present a new protocol for synthesizing the qubit $\ket{+}$ state which requires only Clifford+T gates, and computational basis state preparation and measurement.  By applying one qutrit Clifford+T gate Clifford equivalent to the R gate before measurement, this protocol can be adapted to synthesize the qubit $\ket{-}$ state, which is Clifford equivalent to the $i$ eigenstate of the qutrit Hadamard gate.

\begin{lemma}
    The qubit $\ket{+} = \frac{\ket{0}+\ket{1}}{\sqrt{2}}$ state can be synthesized with 2/3 probability of success using 3 qutrit T gates and qutrit stabilizer states, gates, and measurements, with one ancilla qutrit.
\end{lemma}
\begin{proof}
    The qutrit Clifford+T circuit below produces the qubit $\ket{+}$ state with 2/3 probability.
    \begin{equation}
        \input{qutrit/qubitplus.tikz}
    \end{equation}
    Before measurement, the state is in uniform superposition of $\ket{00}$, $\ket{10}$, and $\ket{21}$; measuring the bottom qubit to be the $\ket{0}$ state then produces the $\frac{\ket{0}+\ket{1}}{\sqrt{2}}$ state, equal to the qubit $\ket{+}$ state.
\end{proof}
\begin{corollary}
    By application of the $X_{12} R X_{12} = \text{diag}(1,-1,1)$ gate before measurement, the qubit $\ket{-} = \frac{\ket{0}-\ket{1}}{\sqrt{2}}$ state can be subsequently attained requiring only Clifford+T gates, and computational basis state preparation and measurement, with one ancilla qutrit.
\end{corollary}

\begin{corollary}
    By Clifford equivalence, the above proposition and corollary respectively construct the Norell state $\frac{\ket{1}+\ket{2}}{\sqrt{2}}$, and the $i$ eigenstate of the qutrit Hadamard gate $\ket{H, i} = \frac{\ket{1} - \ket{2}}{2}$.
\end{corollary}

\subsection{Qutrit Clifford+T Emulation of Multiqubit Toffoli and AND Gates}\label{sec:mctand}
In this subsection, we present qutrit T-counts for ternary classical reversible circuit decompositions of multiple-controlled Toffoli gates, ancilla-free in the qutrit Clifford+T gate set, as presented in Table~\ref{table:countstable}. We note that the gate counts (of either the CCX gate or the $\ket{0}$-controlled X gate, and the qutrit H gate) in the qutrit Toffoli+Hadamard gate set as defined in~\cite{Roy_2023, glaudell_exact_2024} are apparent from the decompositions. We will first present linear depth circuit decompositions, followed by logarithmic depth circuit decompositions--- all with linear gate count.

A qubit Toffoli gate can be built from qutrit gates using the decomposition by Gokhale~et.~al.~\cite{gokhale_asymptotic_2019,baker2020improved}:
\begin{equation}
\begin{array}{c}
     \Qcircuit @R=1.5em @C=0.25em {
    \lstick{\ket{q_{0}}} & \qw & \ctrl{2} & \qw \\
    \lstick{\ket{q_{1}}} & \qw & \ctrl{1} & \qw \\
    \lstick{\ket{q_{2}}} & \qw & \targ & \qw \\
    }
\end{array}
 \, \emulates \qquad
 \begin{array}{c}
 \Qcircuit @R=0.5em @C=0.25em {
    \lstick{\ket{q_{0}}} & \onecontrol & \qw  & \onecontrol & \qw & \\
    \lstick{\ket{q_{1}}} & \gate{X_{+1}}\qwx  & \twocontrol & \gate{X_{-1}}\qwx  &  \qw & \\
    \lstick{\ket{q_{2}}} & \qw  & \gate{X_{01}} \qwx  & \qw  &  \qw & \\
    }
 \end{array}
\label{eq:toffdecomp}
\end{equation}

To exactly emulate the qubit CCX gate, the $X$ gate on the target qutrit must be $X_{01}$.  In their presentation of this gate they simply denoted the gate on the target qutrit as $X$, referring to any of the five single-qutrit computational basis permutation gates.
The action of the gate is irrelevant for the cases where any of the input qutrits are in the $\ket{2}$ state, making it more efficient to use Clifford CX gates to replace the top two controlled gates:
\begin{equation}
    \begin{array}{c}\Qcircuit @R=0.5em @C=0.25em {
    \lstick{\ket{q_{0}}} & \onecontrol & \twocontrol & \qw &\\
    \lstick{\ket{q_{1}}} & \gate{X_{+1}}\qwx  & \gate{X_{-1}}\qwx & \qw &\\
    }\end{array}= \input{qutrit/cxplus.tikz}
    \label{eq:qutritcx}
\end{equation}
\begin{equation}
    \begin{array}{c}\Qcircuit @R=0.5em @C=0.25em {
    \lstick{\ket{q_{0}}} & \onecontrol & \twocontrol & \qw & \\
    \lstick{\ket{q_{1}}} & \gate{X_{-1}}\qwx  & \gate{X_{+1}}\qwx  &  \qw & \\
    }\end{array} = \input{qutrit/cxminus.tikz}
\end{equation}

Without compromising the correctness of the qubit Toffoli gate, the above decomposition in Equation~\eqref{eq:toffdecomp} was modified by Bocharov, Roetteler, and Svore as follows~\cite[Proposition 6]{bocharov_factoring_2017}:
\begin{equation}
\begin{array}{c}
     \Qcircuit @R=1.5em @C=0.25em {
    \lstick{\ket{q_{0}}} & \qw & \ctrl{2} & \qw \\
    \lstick{\ket{q_{1}}} & \qw & \ctrl{1} & \qw \\
    \lstick{\ket{q_{2}}} & \qw & \gate{X} & \qw \\
    }
\end{array}
 \, \emulates \input{qutrit/qubitccx.tikz}
\label{eq:toffdecompcliffordplust}
\end{equation}
They concluded that the qubit CCX gate takes 15 T-gates to emulate with qutrit Clifford+T gates, taking the bottom gate to be the $\ket{2}$-controlled $X_{01}$ gate, as the $X_{01}$ gate emulates the qubit X gate.  They lower the T-count to 12 by adding one clean ancilla.

We show that the T-count can be 12 without ancillae, by $\ket{2}$-controlling an emulation of the qubit X gate instead of the $X_{01}$ gate.
\begin{proposition}
    \label{prop:tcd01qutrit-qubit}
    The $\ket{2}$-controlled qubit X gate, which executes the qubit $X$ gate if and only if the control qutrit is in the $\ket{2}$ state, can be emulated fault-tolerantly by a two-qutrit Clifford+T unitary without ancillae, with a T-count of 12.
\end{proposition}
\begin{proof}
    The decomposition is as follows.  Recall that each of the 4 $\ket{2}$-controlled $X_{\pm 1}$ gates has a T-count of 3.
    \begin{equation}
        \input{qutrit/tcd01qutrit-qubit.tikz}
    \end{equation}
\end{proof}
\begin{corollary}
    By Equation~\eqref{eq:toffdecompcliffordplust}, the qubit CCX gate can be emulated unitarily without ancillae using Clifford+T gates, with T-count 12.
\end{corollary}

For the CCZ gate, from~\cite[Proposition 5.5]{wetering_building_2023}:
\begin{proposition}
    The qubit CCZ gate can be emulated unitarily without ancillae using Clifford+R gates, with R-count 3 and Clifford CX-count 5.
\end{proposition}
\begin{proof}
    Consider the qubit CCX gate emulation in Equation~\eqref{eq:toffdecompcliffordplust}.  If instead the controlled gate on the target is the $\ket{2}$-controlled qubit $Z$ gate in~\cite[Corollary 5.4]{wetering_building_2023}, which has R-count 3 and Clifford CX-count 3, the emulated gate is the qubit CCZ gate.
\end{proof}

The most straightforward way to generalize this to the qubit multiple-controlled Toffoli is the following linear gate count and linear depth decomposition:
\begin{proposition}
    The technique of emulating the qubit Toffoli can be generalized to $n$ controls, and for any unitary U which can be $\ket{2}$-controlled.
\end{proposition}
\begin{proof}
    \begin{equation}
        \label{eq:qubitcnu}
        \input{qutrit/qubitcnu.tikz}
    \end{equation}
    As each $\ket{2}$-controlled $X_{\pm 1}$ gate is comprised of three qutrit T gates, the T-count is $6(n-2)$ plus the cost of implementing the $\ket{2}$-controlled unitary.    
\end{proof}
\begin{corollary}
    As the qubit Toffoli gate $\ket{2}$-controlled qubit X gate can be constructed with T-count 12 according to Proposition~\ref{prop:tcd01qutrit-qubit}, the qubit generalized Toffoli with $n$ controls can be unitarily emulated without ancillae using qutrit Clifford+T gates, with T-count of exactly $6n$.
\end{corollary}
This is an improvement over~\cite[Corollary 9]{bocharov_factoring_2017}, which emulates the qubit CCCNOT gate (i.e. $n = 3$) with T-count 18 using two clean ancillas, or T-count 21 using one clean ancilla.

In logarithmic depth, qubit multiple-controlled Toffoli gate can be built from qutrit Toffoli gates~\cite{gokhale_asymptotic_2019,baker2020improved}.
Gokhale et.~al.~\cite{gokhale_asymptotic_2019} leveraged the $\ket{22}$-controlled gate to construct a qutrit emulation of the qubit Toffoli with $n$ controls, with circuit depth logarithmic in $n$.  We reprint their decomposition in Figure~\ref{fig:logdepthcnuemulation}.  They cited the decomposition for the $\ket{22}$-controlled $X$ gates by Di and Wei, in terms of qubit Clifford+T gates pairwise between the three qutrit computational basis states.  As far as we know, a protocol to simultaneously error correct for these three gate sets has not been formulated.

\begin{figure}[H]
    \begin{subfigure}[b]{0.49\textwidth}
      \tikzfigscale{0.7}{qutrit/qubitcnulog-qcirc}
      \caption{Qubit Generalized Toffoli gate emulation using intermediate qutrits by Gokhale et.~al.~\cite[Figure 5]{gokhale_asymptotic_2019}.}
      \label{fig:logdepthcnuemulation}
    \end{subfigure}
    \begin{subfigure}[b]{0.49\textwidth}
        \tikzfigscale{0.7}{qutrit/qubitcnulog}
        \caption{We modified the left circuit to optimize T-count, given that we presently have a lower T-count decomposition for the $\ket{22}$-controlled $X_{12}$ gate than the $\ket{22}$-controlled $X_{+1}$ gate.}
        \label{fig:logdepthcnuemulationv2}
    \end{subfigure}
  \end{figure}

Using the decomposition for $\ket{22}$-controlling any qutrit unitary $V$ from~\cite[Lemma 1]{yeh_qutrit_ctrl_2022} to build $\ket{22}$-controlled $X_{12}$ in 51 T gates, we adapt their decomposition in Figure~\ref{fig:logdepthcnuemulation} to the qutrit Clifford+T gateset.  We thereby achieve linear qutrit T-count for the qubit generalized Toffoli gate, whilst retaining the desirable properties of logarithmic circuit depth and no logical ancillae needed (i.e. only using those necessary for magic state injection).  We find it to be more efficient in qutrit T-count to modify each of their $\ket{22}$-controlled $X_{+1}$ gates to $\ket{22}$-controlled $X_{12}$ gates, for a constant factor improvement in lowering the T-count.  To further decrease the resource requirements, we note that all boundary (either the inputs or outputs are qubits) $\ket{22}$-controlled gates in this decomposition are overconstrained; substituting them with our $CCX_{+1}$ unitary decomposition in Equation~\eqref{eq:toffdecomp} with T-count 3 suffices.

Despite this decomposition being linear in T-count, the coefficient of linearity is quite high, given that our decomposition for the $\ket{22}$-controlled $X_{12}$ gate has T-count 51.  It is the first fault-tolerant decomposition for it and we have not yet attempted to optimize this, reason being that there is a far more efficient construction using the $n$-ary AND gate familiar from binary classical computing, implemented on superconducting qutrits in~\cite{chu_scalable_2023}, and resynthesized here in qutrit Clifford+T.

\begin{proposition}
    \label{prop:narybinaryand}
    Consider the $n$-ary binary AND gate, which returns 1 iff all binary inputs are 1, and 0 otherwise.  It can be emulated unitarily by Clifford+T gates without ancillae, in log($n$) depth, with T-count $3n-3$ and Clifford CX-count $4n-4$.
\end{proposition}
\begin{proof}
This follows from the equivalent well-known construction in classical reversible computation. The idea is to apply the T-count $3$ emulation of the binary AND gate from Lemma~\ref{lemma:binaryand} $n-1$ times, as a binary tree structure.  This generalizes to any positive integer $n$. The $n=7$ case is shown below:
\begin{equation}
    \input{qutrit/7-arybinaryandud.tikz}
\end{equation}
\end{proof}
At first glance, it may seem that with the exception of the AND'd output, the other outputs can be disregarded.  However, they are necessary for the gate as a whole to be unitary and thus reversible.  At the end of this section, we will demonstrate the usefulness of being able to compute and uncompute the $n$-ary binary AND gate.

This means that the multiple-controlled qubit Toffoli admits the following decomposition when $n$-ary AND is available~\cite{chu_scalable_2023}, for which we computed the T-count:
\begin{theorem}
    The qubit C$^n$ X gate can be emulated unitarily without ancillae in qutrit Clifford+T gates in O(log($n$)) depth, with T-count exactly $6n$ and Clifford CX-count $8n$.
\end{theorem}
\begin{proof}
    First apply the $n$-ary binary AND gate emulation of Proposition~\ref{prop:narybinaryand} with T-count $3(n-1)$ and Clifford CX-count $4(n-1)$, where the bottommost output wire is the AND'd output.  Then, we apply the qubit CX gate emulation of Lemma~\ref{lemma:qubitcx} controlled on the AND'd output, targeting the target of the C$^n$ X gate, with T-count 6 and Clifford CX-count 8.  Finally, we uncompute the $n$-ary binary AND gate by applying its inverse, also with T-count $3(n-1)$ and Clifford CX-count $4(n-1)$.
    \begin{equation}
        \input{qutrit/qubitcnuandud.tikz}
    \end{equation}
    The T-count totals exactly $6n$ and the Clifford CX-count totals exactly $8n$.
\end{proof}
\begin{corollary}
    This construction can be generalized to gates of the form C$^n$ U, controlled on $\ket{1}^{\otimes n}$ for qubit inputs, where U is any target unitary acting on $m$ qudits for which a controlled version, for some type of control, can be built.
\end{corollary}

%% file: figures/qutrit/qubitcqubitx.tikz
\begin{tikzpicture}
	\begin{pgfonlayer}{nodelayer}
		\node [style=small dot] (0) at (-7, 1) {};
		\node [style=targ] (1) at (-7, -1) {$+$};
		\node [style=none] (2) at (-5.5, -0.125) {$=$};
		\node [style=vertex set] (3) at (-3.25, -1) {$\Lambda$};
		\node [style=box] (4) at (-3.25, 1) {$X_{-1}$};
		\node [style=box] (5) at (-1, -1) {$X_{12}$};
		\node [style=box] (6) at (-1, 1) {$X_{12}$};
		\node [style=1 control] (7) at (1.25, 1) {\textsf{1}};
		\node [style=box] (8) at (1.25, -1) {$X_{-1}$};
		\node [style=1 control] (9) at (3.5, -1) {\textsf{1}};
		\node [style=box] (10) at (3.5, 1) {$X_{+1}$};
		\node [style=none] (11) at (6.5, 1) {};
		\node [style=none] (12) at (4.75, 1) {};
		\node [style=none] (13) at (4.75, -1) {};
		\node [style=none] (14) at (6.5, -1) {};
		\node [style=box] (15) at (7.75, -1) {$X_{12}$};
		\node [style=box] (16) at (7.75, 1) {$X_{12}$};
		\node [style=vertex set] (17) at (10, -1) {$\Lambda$};
		\node [style=box] (18) at (10, 1) {$X_{+1}$};
		\node [style=none] (19) at (11.5, 1) {};
		\node [style=none] (20) at (11.5, -1) {};
		\node [style=none] (21) at (-4.75, 1) {};
		\node [style=none] (22) at (-4.75, -1) {};
		\node [style=none] (23) at (-6.25, 1) {};
		\node [style=none] (24) at (-6.25, -1) {};
		\node [style=none] (25) at (-7.75, 1) {};
		\node [style=none] (26) at (-7.75, -1) {};
	\end{pgfonlayer}
	\begin{pgfonlayer}{edgelayer}
		\draw (0) to (1);
		\draw [style=none] (3) to (4);
		\draw [style=none] (7) to (8);
		\draw [style=none] (9) to (10);
		\draw [in=0, out=-180] (14.center) to (12.center);
		\draw [in=-180, out=0] (13.center) to (11.center);
		\draw [style=none] (17) to (18);
		\draw [style=none] (21.center) to (4);
		\draw [style=none] (4) to (6);
		\draw [style=none] (6) to (7);
		\draw [style=none] (7) to (10);
		\draw [style=none] (10) to (12.center);
		\draw [style=none] (11.center) to (16);
		\draw [style=none] (16) to (18);
		\draw [style=none] (18) to (19.center);
		\draw [style=none] (20.center) to (17);
		\draw [style=none] (17) to (15);
		\draw [style=none] (15) to (14.center);
		\draw [style=none] (13.center) to (9);
		\draw [style=none] (9) to (8);
		\draw [style=none] (8) to (5);
		\draw [style=none] (5) to (3);
		\draw [style=none] (3) to (22.center);
		\draw [style=none] (25.center) to (0);
		\draw [style=none] (0) to (23.center);
		\draw [style=none] (26.center) to (1);
		\draw [style=none] (1) to (24.center);
	\end{pgfonlayer}
\end{tikzpicture}

%% file: figures/qutrit/qubitcz.tikz
\begin{tikzpicture}
	\begin{pgfonlayer}{nodelayer}
		\node [style=none] (0) at (-3.5, 1) {};
		\node [style=none] (1) at (-3.5, -1) {};
		\node [style=none] (2) at (-5, 1) {};
		\node [style=none] (3) at (-5, -1) {};
		\node [style=small dot] (4) at (-4.25, 1) {};
		\node [style=small dot] (5) at (-4.25, -1) {};
		\node [style=none] (20) at (-2.5, 0) {$\emulates$};
		\node [style=none] (37) at (4, 1) {};
		\node [style=none] (38) at (4, -1) {};
		\node [style=vertex set] (45) at (-0.375, 1) {$\Lambda$};
		\node [style=box] (46) at (-0.375, -1) {$X_{+1}$};
		\node [style=none] (49) at (-1.5, 1) {};
		\node [style=none] (50) at (-1.5, -1) {};
		\node [style=vertex set] (54) at (3, 1) {$\Lambda$};
		\node [style=box] (55) at (3, -1) {$X_{-1}$};
		\node [style=box] (56) at (1.25, -1) {$R$};
	\end{pgfonlayer}
	\begin{pgfonlayer}{edgelayer}
		\draw [style=none, in=180, out=0] (2.center) to (4);
		\draw [style=none] (4) to (0.center);
		\draw [style=none] (3.center) to (5);
		\draw [style=none] (5) to (1.center);
		\draw [style=none] (4) to (5);
		\draw [style=none] (45) to (46);
		\draw [style=none] (54) to (55);
		\draw [style=none] (49.center) to (45);
		\draw [style=none] (50.center) to (46);
		\draw (45) to (54);
		\draw (54) to (37.center);
		\draw (38.center) to (55);
		\draw (46) to (56);
		\draw (56) to (55);
	\end{pgfonlayer}
\end{tikzpicture}

%% file: figures/qutrit/qubitplus.tikz
\begin{tikzpicture}
	\begin{pgfonlayer}{nodelayer}
		\node [style=none] (0) at (-6, 0) {};
		\node [style=none] (1) at (-6.625, 0) {$\ket{+}$};
		\node [style=2 control] (4) at (3, 0) {\textsf{2}};
		\node [style=box] (5) at (3, -2) {$X_{+1}$};
		\node [style=none] (8) at (6.5, 0) {};
		\node [style=none] (12) at (-5, 0) {};
		\node [style=none] (17) at (-3.25, 0) {$\emulates$};
		\node [style=none] (18) at (-0.5, -2) {};
		\node [style=none] (19) at (-1.125, -2) {$\ket{0}$};
		\node [style=none] (21) at (-0.5, 0) {};
		\node [style=none] (22) at (-1.125, 0) {$\ket{0}$};
		\node [style=box] (23) at (1, 0) {$H$};
		\node [style=none] (24) at (5, -2) {};
		\node [style=none] (25) at (5.625, -2) {$\bra{0}$};
	\end{pgfonlayer}
	\begin{pgfonlayer}{edgelayer}
		\draw (4) to (5);
		\draw (4) to (8.center);
		\draw [style=none] (0.center) to (12.center);
		\draw [style=none] (21.center) to (23);
		\draw [style=none] (23) to (4);
		\draw [style=none] (18.center) to (5);
		\draw [style=none] (5) to (24.center);
	\end{pgfonlayer}
\end{tikzpicture}

%% file: figures/qutrit/cxplus.tikz
\begin{tikzpicture}
	\begin{pgfonlayer}{nodelayer}
		\node [style=none] (2) at (-2.25, 0.75) {};
		\node [style=none] (3) at (-2.25, -0.75) {};
		\node [style=none] (4) at (0.5, -0.75) {};
		\node [style=none] (5) at (0.5, 0.75) {};
		\node [style=none] (14) at (-2.75, 0) {};
		\node [style=vertex set] (15) at (-0.875, 0.75) {$\Lambda$};
		\node [style=box] (16) at (-0.875, -0.75) {$X_{+1}$};
	\end{pgfonlayer}
	\begin{pgfonlayer}{edgelayer}
		\draw [style=none] (15) to (16);
		\draw [style=none] (2.center) to (15);
		\draw [style=none] (15) to (5.center);
		\draw [style=none] (4.center) to (16);
		\draw [style=none] (16) to (3.center);
	\end{pgfonlayer}
\end{tikzpicture}

%% file: figures/qutrit/cxminus.tikz
\begin{tikzpicture}
	\begin{pgfonlayer}{nodelayer}
		\node [style=none] (15) at (-2.25, 0.75) {};
		\node [style=none] (16) at (-2.25, -0.75) {};
		\node [style=none] (17) at (0.5, -0.75) {};
		\node [style=none] (18) at (0.5, 0.75) {};
		\node [style=none] (19) at (-2.75, 0) {};
		\node [style=vertex set] (20) at (-0.875, 0.75) {$\Lambda$};
		\node [style=box] (21) at (-0.875, -0.75) {$X_{-1}$};
	\end{pgfonlayer}
	\begin{pgfonlayer}{edgelayer}
		\draw [style=none] (20) to (21);
		\draw [style=none] (15.center) to (20);
		\draw [style=none] (20) to (18.center);
		\draw [style=none] (17.center) to (21);
		\draw [style=none] (21) to (16.center);
	\end{pgfonlayer}
\end{tikzpicture}

%% file: figures/qutrit/qubitccx.tikz
\begin{tikzpicture}
	\begin{pgfonlayer}{nodelayer}
		\node [style=none] (2) at (-2.25, 1.5) {};
		\node [style=none] (3) at (-2.25, 0) {};
		\node [style=none] (14) at (-2.75, 0) {};
		\node [style=vertex set] (15) at (-0.875, 1.5) {$\Lambda$};
		\node [style=box] (16) at (-0.875, 0) {$X_{+1}$};
		\node [style=none] (18) at (-2.25, -1.5) {};
		\node [style=none] (19) at (5, -1.5) {};
		\node [style=none] (20) at (3.625, 0) {};
		\node [style=2 control] (21) at (1.375, 0) {\textsf{2}};
		\node [style=box] (22) at (1.375, -1.5) {$X$};
		\node [style=none] (25) at (5, 1.5) {};
		\node [style=vertex set] (26) at (3.625, 1.5) {$\Lambda$};
		\node [style=box] (27) at (3.625, 0) {$X_{-1}$};
		\node [style=none] (28) at (5, 0) {};
	\end{pgfonlayer}
	\begin{pgfonlayer}{edgelayer}
		\draw [style=none] (15) to (16);
		\draw [style=none] (2.center) to (15);
		\draw [style=none] (16) to (3.center);
		\draw [style=none] (21) to (22);
		\draw [style=none, in=180, out=0] (21) to (20.center);
		\draw [style=none] (19.center) to (22);
		\draw [style=none] (22) to (18.center);
		\draw [style=none] (26) to (27);
		\draw [style=none] (26) to (25.center);
		\draw [style=none] (15) to (26);
		\draw [style=none] (16) to (21);
		\draw [style=none] (27) to (28.center);
	\end{pgfonlayer}
\end{tikzpicture}

%% file: figures/qutrit/tcd01qutrit-qubit.tikz
\begin{tikzpicture}
	\begin{pgfonlayer}{nodelayer}
		\node [style=2 control] (0) at (-11.75, 1) {\textsf{2}};
		\node [style=box] (1) at (-11.75, -1) {$X_{01}$};
		\node [style=none] (3) at (-9.75, 0) {=};
		\node [style=none] (14) at (6.25, -1) {};
		\node [style=none] (19) at (-8.75, 1) {};
		\node [style=none] (20) at (-8.75, -1) {};
		\node [style=none] (21) at (6.25, 1) {};
		\node [style=none] (26) at (-12.75, 1) {};
		\node [style=none] (27) at (-12.75, -1) {};
		\node [style=none] (28) at (-10.75, 1) {};
		\node [style=none] (29) at (-10.75, -1) {};
		\node [style=box] (31) at (-7.25, 1) {$X_{12}$};
		\node [style=box] (32) at (4.75, 1) {$X_{12}$};
		\node [style=2 control] (33) at (-5.25, 1) {\textsf{2}};
		\node [style=box] (34) at (-5.25, -1) {$X_{-1}$};
		\node [style=1 control] (35) at (-3.25, -1) {\textsf{1}};
		\node [style=box] (36) at (-3.25, 1) {$X_{-1}$};
		\node [style=1 control] (37) at (-1.25, 1) {\textsf{1}};
		\node [style=box] (38) at (-1.25, -1) {$X_{-1}$};
		\node [style=1 control] (39) at (0.75, -1) {\textsf{1}};
		\node [style=box] (40) at (0.75, 1) {$X_{+1}$};
		\node [style=box] (41) at (2.75, -1) {$X_{-1}$};
		\node [style=vertex set] (42) at (2.75, 1.025) {$\Lambda$};
	\end{pgfonlayer}
	\begin{pgfonlayer}{edgelayer}
		\draw [in=90, out=-90] (0) to (1);
		\draw (26.center) to (0);
		\draw (0) to (28.center);
		\draw (27.center) to (1);
		\draw (1) to (29.center);
		\draw (42) to (41);
		\draw [style=none] (19.center) to (31);
		\draw [style=none] (31) to (33);
		\draw [style=none] (33) to (36);
		\draw [style=none] (36) to (37);
		\draw [style=none] (37) to (40);
		\draw [style=none] (40) to (42);
		\draw [style=none] (42) to (32);
		\draw [style=none] (32) to (21.center);
		\draw [style=none] (14.center) to (41);
		\draw [style=none] (41) to (39);
		\draw [style=none] (39) to (38);
		\draw [style=none] (38) to (35);
		\draw [style=none] (35) to (34);
		\draw [style=none] (34) to (20.center);
		\draw [style=none] (33) to (34);
		\draw [style=none] (35) to (36);
		\draw [style=none] (37) to (38);
		\draw [style=none] (39) to (40);
	\end{pgfonlayer}
\end{tikzpicture}

%% file: figures/qutrit/qubitcnu.tikz
\begin{tikzpicture}
	\begin{pgfonlayer}{nodelayer}
		\node [style=none] (3) at (-12, -1.5) {};
		\node [style=2 control] (4) at (-7.25, 0) {\textsf{2}};
		\node [style=box] (5) at (-7.25, -1.5) {$X_{+1}$};
		\node [style=none] (8) at (-12, 3) {};
		\node [style=none] (10) at (-6, -1.5) {};
		\node [style=none] (12) at (-6, 3) {};
		\node [style=none] (14) at (-12, 0) {};
		\node [style=2 control] (15) at (-9, 1.5) {\textsf{2}};
		\node [style=box] (16) at (-9, 0) {$X_{+1}$};
		\node [style=none] (17) at (-12, 1.5) {};
		\node [style=none] (18) at (-6, 0) {};
		\node [style=none] (19) at (-6, 1.5) {};
		\node [style=none] (22) at (1.5, -1.5) {};
		\node [style=2 control] (23) at (-3.25, 0) {\textsf{2}};
		\node [style=box] (24) at (-3.25, -1.5) {$X_{-1}$};
		\node [style=none] (25) at (1.5, 3) {};
		\node [style=none] (26) at (-4.5, -1.5) {};
		\node [style=none] (27) at (-4.5, 3) {};
		\node [style=none] (28) at (1.5, 0) {};
		\node [style=2 control] (29) at (-1.5, 1.5) {\textsf{2}};
		\node [style=box] (30) at (-1.5, 0) {$X_{-1}$};
		\node [style=none] (31) at (1.5, 1.5) {};
		\node [style=none] (32) at (-4.5, 0) {};
		\node [style=none] (33) at (-4.5, 1.5) {};
		\node [style=2 control] (36) at (-5.25, -3.25) {\textsf{2}};
		\node [style=none] (37) at (-12, -4.75) {};
		\node [style=none] (38) at (-12, -3.25) {};
		\node [style=none] (39) at (1.5, -4.75) {};
		\node [style=none] (40) at (1.5, -3.25) {};
		\node [style=box] (41) at (-5.25, -4.75) {$U$};
		\node [style=none] (42) at (-7.5, -4.5) {};
		\node [style=none] (43) at (-7.75, -5) {};
		\node [style=none] (44) at (-2.5, -4.5) {};
		\node [style=none] (45) at (-2.75, -5) {};
		\node [style=none] (46) at (-7.5, -4.25) {$m$};
		\node [style=none] (47) at (-2.5, -4.25) {$m$};
		\node [style=none, rotate=-45] (48) at (-6, -2.5) {$\dots$};
		\node [style=none, rotate=45] (49) at (-4.5, -2.5) {$\dots$};
		\node [style=none] (51) at (-17, 3) {};
		\node [style=none] (52) at (-17, 0) {};
		\node [style=none] (53) at (-17, 1.5) {};
		\node [style=none] (54) at (-17, -4.75) {};
		\node [style=none] (55) at (-17, -3.25) {};
		\node [style=none] (57) at (-14, 3) {};
		\node [style=none] (58) at (-14, 0) {};
		\node [style=none] (59) at (-14, 1.5) {};
		\node [style=none] (60) at (-14, -4.75) {};
		\node [style=none] (61) at (-14, -3.25) {};
		\node [style=small dot] (62) at (-15.5, 3) {};
		\node [style=small dot] (63) at (-15.5, 1.5) {};
		\node [style=small dot] (64) at (-15.5, 0) {};
		\node [style=small dot] (65) at (-15.5, -3.25) {};
		\node [style=box] (66) at (-15.5, -4.75) {$U$};
		\node [style=none] (67) at (-16.25, -4.5) {};
		\node [style=none] (68) at (-16.5, -5) {};
		\node [style=none] (69) at (-16.25, -4.25) {$m$};
		\node [style=none] (70) at (-14.5, -4.5) {};
		\node [style=none] (71) at (-14.75, -5) {};
		\node [style=none] (72) at (-14.5, -4.25) {$m$};
		\node [style=none] (73) at (-15.5, -0.75) {};
		\node [style=none] (74) at (-15.5, -2.5) {};
		\node [style=none] (75) at (-15.5, -1.5) {$\vdots$};
		\node [style=none] (76) at (-13, -0.875) {$\emulates$};
		\node [style=vertex set] (82) at (-10.625, 3) {$\Lambda$};
		\node [style=box] (83) at (-10.625, 1.5) {$X_{+1}$};
		\node [style=vertex set] (84) at (0.125, 3) {$\Lambda$};
		\node [style=box] (85) at (0.125, 1.5) {$X_{-1}$};
		\node [style=none, rotate=-45] (86) at (-5.25, 0.75) {\rotatebox{45}{$\dots$}};
	\end{pgfonlayer}
	\begin{pgfonlayer}{edgelayer}
		\draw [style=none] (3.center) to (5);
		\draw [style=none, in=90, out=-90] (4) to (5);
		\draw [style=none] (5) to (10.center);
		\draw [style=none] (14.center) to (16);
		\draw [style=none, in=90, out=-90] (15) to (16);
		\draw [style=none] (15) to (19.center);
		\draw [style=none] (4) to (18.center);
		\draw (16) to (4);
		\draw [style=none] (22.center) to (24);
		\draw [style=none, in=90, out=-90] (23) to (24);
		\draw [style=none] (24) to (26.center);
		\draw [style=none] (28.center) to (30);
		\draw [style=none, in=90, out=-90] (29) to (30);
		\draw [style=none] (29) to (33.center);
		\draw [style=none] (23) to (32.center);
		\draw (30) to (23);
		\draw (38.center) to (36);
		\draw (36) to (40.center);
		\draw (36) to (41);
		\draw (37.center) to (41);
		\draw (41) to (39.center);
		\draw (44.center) to (45.center);
		\draw (43.center) to (42.center);
		\draw (65) to (66);
		\draw (68.center) to (67.center);
		\draw (71.center) to (70.center);
		\draw (55.center) to (65);
		\draw (65) to (61.center);
		\draw (60.center) to (66);
		\draw (66) to (54.center);
		\draw (52.center) to (64);
		\draw (64) to (58.center);
		\draw (59.center) to (63);
		\draw (63) to (53.center);
		\draw (62) to (51.center);
		\draw (62) to (57.center);
		\draw (62) to (63);
		\draw (63) to (64);
		\draw (73.center) to (64);
		\draw (74.center) to (65);
		\draw [style=none] (82) to (83);
		\draw [style=none] (84) to (85);
		\draw [style=none] (29) to (85);
		\draw [style=none] (85) to (31.center);
		\draw [style=none] (25.center) to (84);
		\draw [style=none] (84) to (27.center);
		\draw [style=none] (12.center) to (82);
		\draw [style=none] (82) to (8.center);
		\draw [style=none] (17.center) to (83);
		\draw [style=none] (83) to (15);
	\end{pgfonlayer}
\end{tikzpicture}

%% file: figures/qutrit/7-arybinaryandud.tikz
\begin{tikzpicture}
	\begin{pgfonlayer}{nodelayer}
		\node [style=none] (0) at (1.25, -9.625) {};
		\node [style=none] (3) at (1.25, -7.875) {};
		\node [style=none] (29) at (0, -4.375) {$\emulates$};
		\node [style=none] (30) at (-6.5, 0.875) {};
		\node [style=none] (31) at (-3, -9.625) {};
		\node [style=none] (32) at (-6.5, -0.875) {};
		\node [style=none] (33) at (-3, -7.875) {};
		\node [style=none] (38) at (-6, 0.875) {};
		\node [style=none] (39) at (-3.5, -9.625) {};
		\node [style=none] (40) at (-6, -0.875) {};
		\node [style=none] (41) at (-3.5, -7.875) {};
		\node [style=2 control] (42) at (4, -7.875) {\textsf{2}};
		\node [style=box] (43) at (4, -9.625) {$X_{-1}$};
		\node [style=vertex set] (44) at (2.5, -9.625) {$\Lambda$};
		\node [style=box] (45) at (2.5, -7.875) {$X_{-1}$};
		\node [style=none] (46) at (-7, 0.875) {q0};
		\node [style=none] (47) at (-7, -0.875) {q1};
		\node [style=none] (48) at (-1.5, -9.625) {q0$\wedge$...$\wedge$q6};
		\node [style=none] (49) at (-6.5, -2.625) {};
		\node [style=none] (50) at (-3, -6.125) {};
		\node [style=none] (51) at (-6.5, -4.375) {};
		\node [style=none] (52) at (-3, -4.375) {};
		\node [style=none] (53) at (-4.75, -3.875) {7-ary};
		\node [style=none] (54) at (-6, -2.625) {};
		\node [style=none] (55) at (-3.5, -6.125) {};
		\node [style=none] (56) at (-6, -4.375) {};
		\node [style=none] (57) at (-3.5, -4.375) {};
		\node [style=none] (58) at (-7, -2.625) {q2};
		\node [style=none] (59) at (-7, -4.375) {q3};
		\node [style=none] (62) at (-6.5, -6.125) {};
		\node [style=none] (63) at (-3, -2.625) {};
		\node [style=none] (64) at (-6, -6.125) {};
		\node [style=none] (65) at (-3.5, -2.625) {};
		\node [style=none] (66) at (-7, -6.125) {q4};
		\node [style=none] (67) at (-6.5, -7.875) {};
		\node [style=none] (68) at (-3, -0.875) {};
		\node [style=none] (69) at (-6.5, -9.625) {};
		\node [style=none] (70) at (-3, 0.875) {};
		\node [style=none] (72) at (-6, -7.875) {};
		\node [style=none] (73) at (-3.5, -0.875) {};
		\node [style=none] (74) at (-6, -9.625) {};
		\node [style=none] (75) at (-3.5, 0.875) {};
		\node [style=none] (76) at (-7, -7.875) {q5};
		\node [style=none] (77) at (-7, -9.625) {q6};
		\node [style=none] (78) at (1.25, -6.125) {};
		\node [style=none] (81) at (1.25, -4.375) {};
		\node [style=2 control] (82) at (4, -4.375) {\textsf{2}};
		\node [style=box] (83) at (4, -6.125) {$X_{-1}$};
		\node [style=vertex set] (84) at (2.5, -6.125) {$\Lambda$};
		\node [style=box] (85) at (2.5, -4.375) {$X_{-1}$};
		\node [style=none] (86) at (1.25, -2.625) {};
		\node [style=none] (89) at (1.25, -0.875) {};
		\node [style=2 control] (90) at (4, -0.875) {\textsf{2}};
		\node [style=box] (91) at (4, -2.625) {$X_{-1}$};
		\node [style=vertex set] (92) at (2.5, -2.625) {$\Lambda$};
		\node [style=box] (93) at (2.5, -0.875) {$X_{-1}$};
		\node [style=2 control] (98) at (8, -6.125) {\textsf{2}};
		\node [style=box] (99) at (8, -9.625) {$X_{-1}$};
		\node [style=vertex set] (100) at (6.5, -9.625) {$\Lambda$};
		\node [style=box] (101) at (6.5, -6.125) {$X_{-1}$};
		\node [style=2 control] (106) at (8, 0.875) {\textsf{2}};
		\node [style=box] (107) at (8, -2.625) {$X_{-1}$};
		\node [style=vertex set] (108) at (6.5, -2.625) {$\Lambda$};
		\node [style=box] (109) at (6.5, 0.875) {$X_{-1}$};
		\node [style=2 control] (114) at (12, -2.625) {\textsf{2}};
		\node [style=box] (115) at (12, -9.625) {$X_{-1}$};
		\node [style=vertex set] (116) at (10.5, -9.625) {$\Lambda$};
		\node [style=box] (117) at (10.5, -2.625) {$X_{-1}$};
		\node [style=none] (118) at (-6, 1.25) {};
		\node [style=none] (119) at (-3.5, 1.25) {};
		\node [style=none] (120) at (-6, -10) {};
		\node [style=none] (121) at (-3.5, -10) {};
		\node [style=none] (122) at (-4.75, -4.75) {AND};
		\node [style=none] (123) at (1.25, 0.875) {};
		\node [style=none] (125) at (13.25, -9.625) {};
		\node [style=none] (126) at (13.25, -7.875) {};
		\node [style=none] (127) at (13.25, -6.125) {};
		\node [style=none] (128) at (13.25, -4.375) {};
		\node [style=none] (129) at (13.25, -2.625) {};
		\node [style=none] (130) at (13.25, -0.875) {};
		\node [style=none] (131) at (13.25, 0.875) {};
	\end{pgfonlayer}
	\begin{pgfonlayer}{edgelayer}
		\draw [style=none] (30.center) to (38.center);
		\draw [style=none] (32.center) to (40.center);
		\draw [style=none] (39.center) to (31.center);
		\draw [style=none] (41.center) to (33.center);
		\draw (42) to (43);
		\draw [style=none] (44) to (45);
		\draw [style=none] (3.center) to (45);
		\draw [style=none] (45) to (42);
		\draw [style=none] (43) to (44);
		\draw [style=none] (44) to (0.center);
		\draw [style=none] (49.center) to (54.center);
		\draw [style=none] (51.center) to (56.center);
		\draw [style=none] (55.center) to (50.center);
		\draw [style=none] (57.center) to (52.center);
		\draw [style=none] (62.center) to (64.center);
		\draw [style=none] (65.center) to (63.center);
		\draw [style=none] (67.center) to (72.center);
		\draw [style=none] (69.center) to (74.center);
		\draw [style=none] (73.center) to (68.center);
		\draw [style=none] (75.center) to (70.center);
		\draw (82) to (83);
		\draw [style=none] (84) to (85);
		\draw [style=none] (81.center) to (85);
		\draw [style=none] (85) to (82);
		\draw [style=none] (83) to (84);
		\draw [style=none] (84) to (78.center);
		\draw (90) to (91);
		\draw [style=none] (92) to (93);
		\draw [style=none] (89.center) to (93);
		\draw [style=none] (93) to (90);
		\draw [style=none] (91) to (92);
		\draw [style=none] (92) to (86.center);
		\draw (98) to (99);
		\draw [style=none] (100) to (101);
		\draw [style=none] (101) to (98);
		\draw [style=none] (99) to (100);
		\draw (106) to (107);
		\draw [style=none] (108) to (109);
		\draw [style=none] (109) to (106);
		\draw [style=none] (107) to (108);
		\draw (114) to (115);
		\draw [style=none] (116) to (117);
		\draw [style=none] (117) to (114);
		\draw [style=none] (115) to (116);
		\draw [style=filllayerwhite] (120.center)
			 to (121.center)
			 to (119.center)
			 to (118.center)
			 to cycle;
		\draw [style=none] (125.center) to (115);
		\draw [style=none] (42) to (126.center);
		\draw [style=none] (82) to (128.center);
		\draw [style=none] (90) to (130.center);
		\draw [style=none] (98) to (127.center);
		\draw [style=none] (106) to (131.center);
		\draw [style=none] (114) to (129.center);
		\draw [style=none] (123.center) to (109);
		\draw [style=none] (43) to (100);
		\draw [style=none] (101) to (83);
		\draw [style=none] (91) to (108);
		\draw [style=none] (107) to (117);
		\draw [style=none] (116) to (99);
	\end{pgfonlayer}
\end{tikzpicture}

%% file: figures/qutrit/qubitcnuandud.tikz
\begin{tikzpicture}
	\begin{pgfonlayer}{nodelayer}
		\node [style=none] (8) at (-12, 3) {};
		\node [style=none] (14) at (-12, 0) {};
		\node [style=none] (17) at (-12, 1.5) {};
		\node [style=none] (22) at (-2, -1.5) {};
		\node [style=none] (25) at (-2, 3) {};
		\node [style=none] (28) at (-2, 0) {};
		\node [style=none] (31) at (-2, 1.5) {};
		\node [style=small dot] (36) at (-7, -3.25) {};
		\node [style=none] (37) at (-12, -4.75) {};
		\node [style=none] (38) at (-12, -3.25) {};
		\node [style=none] (39) at (-2, -4.75) {};
		\node [style=none] (40) at (-2, -3.25) {};
		\node [style=targ] (41) at (-7, -4.75) {$+$};
		\node [style=none] (51) at (-16, 3) {};
		\node [style=none] (52) at (-16, 0) {};
		\node [style=none] (53) at (-16, 1.5) {};
		\node [style=none] (54) at (-16, -4.75) {};
		\node [style=none] (55) at (-16, -3.25) {};
		\node [style=none] (57) at (-14, 3) {};
		\node [style=none] (58) at (-14, 0) {};
		\node [style=none] (59) at (-14, 1.5) {};
		\node [style=none] (60) at (-14, -4.75) {};
		\node [style=none] (61) at (-14, -3.25) {};
		\node [style=small dot] (62) at (-15, 3) {};
		\node [style=small dot] (63) at (-15, 1.5) {};
		\node [style=small dot] (64) at (-15, 0) {};
		\node [style=small dot] (65) at (-15, -3.25) {};
		\node [style=targ] (66) at (-15, -4.75) {$+$};
		\node [style=none] (73) at (-15, -0.75) {};
		\node [style=none] (74) at (-15, -2.5) {};
		\node [style=none] (75) at (-15, -1.5) {$\vdots$};
		\node [style=none] (76) at (-13, -0.875) {$\emulates$};
		\node [style=none] (77) at (-12, 3) {};
		\node [style=none] (79) at (-11.25, 3) {};
		\node [style=none] (80) at (-8.25, -3.25) {};
		\node [style=none] (81) at (-9.75, 0.5) {$n$-ary};
		\node [style=none] (82) at (-12, -3.25) {};
		\node [style=none] (84) at (-11.25, -3.25) {};
		\node [style=none] (85) at (-8.25, 3) {};
		\node [style=none] (86) at (-11.25, 3.375) {};
		\node [style=none] (87) at (-8.25, 3.375) {};
		\node [style=none] (88) at (-11.25, -3.625) {};
		\node [style=none] (89) at (-8.25, -3.625) {};
		\node [style=none] (90) at (-9.75, -0.5) {AND};
		\node [style=none] (91) at (-11.25, 1.5) {};
		\node [style=none] (92) at (-11.25, 0) {};
		\node [style=none] (93) at (-11.75, -1.5) {$\vdots$};
		\node [style=none] (94) at (-8.25, 1.5) {};
		\node [style=none] (95) at (-8.25, 0) {};
		\node [style=none] (98) at (-7.75, -1.5) {$\vdots$};
		\node [style=none] (99) at (-2, 3) {};
		\node [style=none] (100) at (-2, 0) {};
		\node [style=none] (101) at (-2, 1.5) {};
		\node [style=none] (102) at (-2, -3.25) {};
		\node [style=none] (103) at (-2, 3) {};
		\node [style=none] (105) at (-2.75, 3) {};
		\node [style=none] (106) at (-5.75, -3.25) {};
		\node [style=none] (107) at (-4.25, 0.5) {$n$-ary};
		\node [style=none] (108) at (-2, -3.25) {};
		\node [style=none] (110) at (-2.75, -3.25) {};
		\node [style=none] (111) at (-5.75, 3) {};
		\node [style=none] (112) at (-2.75, 3.375) {};
		\node [style=none] (113) at (-5.75, 3.375) {};
		\node [style=none] (114) at (-2.75, -3.625) {};
		\node [style=none] (115) at (-5.75, -3.625) {};
		\node [style=none] (116) at (-4.25, -0.5) {AND$^{\dagger}$};
		\node [style=none] (117) at (-2.75, 1.5) {};
		\node [style=none] (118) at (-2.75, 0) {};
		\node [style=none] (119) at (-2.25, -1.5) {$\vdots$};
		\node [style=none] (120) at (-5.75, 1.5) {};
		\node [style=none] (121) at (-5.75, 0) {};
		\node [style=none] (124) at (-6.25, -1.5) {$\vdots$};
		\node [style=none] (125) at (-17, 3) {};
		\node [style=none] (126) at (-17, -3.25) {};
		\node [style=none] (127) at (-17.5, -0.125) {};
		\node [style=none] (128) at (-17.75, -0.25) {};
		\node [style=none] (129) at (-17.75, -0.25) {$n$};
		\node [style=none] (130) at (-12, -1.5) {};
	\end{pgfonlayer}
	\begin{pgfonlayer}{edgelayer}
		\draw [style=none, in=-180, out=0] (85.center) to (111.center);
		\draw [style=none] (95.center) to (121.center);
		\draw [style=none] (94.center) to (120.center);
		\draw (36) to (41);
		\draw (37.center) to (41);
		\draw (41) to (39.center);
		\draw (65) to (66);
		\draw (55.center) to (65);
		\draw (65) to (61.center);
		\draw (60.center) to (66);
		\draw (66) to (54.center);
		\draw (52.center) to (64);
		\draw (64) to (58.center);
		\draw (59.center) to (63);
		\draw (63) to (53.center);
		\draw (62) to (51.center);
		\draw (62) to (57.center);
		\draw (62) to (63);
		\draw (63) to (64);
		\draw (73.center) to (64);
		\draw (74.center) to (65);
		\draw [style=none] (77.center) to (79.center);
		\draw [style=none] (82.center) to (84.center);
		\draw [style=filllayerwhite] (89.center)
			 to (87.center)
			 to (86.center)
			 to (88.center)
			 to cycle;
		\draw [style=none] (17.center) to (91.center);
		\draw [style=none] (14.center) to (92.center);
		\draw [style=none] (103.center) to (105.center);
		\draw [style=none] (108.center) to (110.center);
		\draw [style=none] (101.center) to (117.center);
		\draw [style=none] (100.center) to (118.center);
		\draw [style=none] (80.center) to (36);
		\draw [style=none] (36) to (106.center);
		\draw [in=180, out=0, looseness=0.25] (127.center) to (126.center);
		\draw [in=-180, out=0, looseness=0.25] (127.center) to (125.center);
		\draw [style=filllayerwhite] (112.center)
			 to (113.center)
			 to (115.center)
			 to (114.center)
			 to cycle;
	\end{pgfonlayer}
\end{tikzpicture}